\documentclass{article}

\usepackage{color}
\usepackage[dvipsnames]{xcolor}

\usepackage{moreverb}								
\usepackage{textcomp}								
\usepackage{lmodern}								
\usepackage{helvet}									
\usepackage[T1]{fontenc}							
\usepackage[english]{babel}							
\usepackage[utf8]{inputenc}							
\usepackage{amsmath}								
\usepackage{amssymb}								
\usepackage{graphicx}								
\usepackage{subfig}									
\usepackage{frame}									
\numberwithin{equation}{section}					
\numberwithin{figure}{section}						
\numberwithin{table}{section}						
\usepackage{xspace}

\usepackage[labelfont=bf, textfont=normal,
			justification=justified,
			singlelinecheck=false]{caption}

\usepackage{hyperref}								
\hypersetup{colorlinks, citecolor=black,
   		 	filecolor=black, linkcolor=black,
    		urlcolor=black}

\usepackage{titlesec}		
\titleformat{\chapter}[display]
  {\Huge\bfseries\filcenter}
  {{\fontsize{50pt}{1em}\vspace{-4.2ex}\selectfont \textnormal{\thechapter}}}{1ex}{}[]

\usepackage[textsize=tiny, disable]{todonotes}   
\usepackage{latexsym,amsopn,amsfonts}
\usepackage[boxed,ruled,vlined,linesnumbered]{algorithm2e}

\newcommand{\VCalgSize}{footnotesize}

\newtheorem{requirement}{Requirement}[section]

\newtheorem{theorem}{Theorem}[section]
\newtheorem{lemma}[theorem]{Lemma}

\newtheorem{definition}{Definition}[section]

\newtheorem{assumption}[theorem]{Assumption}
\newtheorem{corollary}[theorem]{Corollary}

\newcommand{\true}{\mathsf{True}\xspace}
\newcommand{\false}{\mathsf{False}\xspace}

\newcommand{\noPriorityViolation}{\mathsf{noPriorityViolation}\xspace}
\newcommand{\noPriorityViolationStat}{\mathsf{noPriorityViolationStrict}\xspace}
\newcommand{\noPriorityViolationDyn}{\mathsf{noPriorityViolationDynamic}\xspace}
\newcommand{\hasLeftCriticalSection}{\mathsf{hasLeftCriticalSection}\xspace}
\newcommand{\ntfyEstimat}{\mathsf{ntfyEstimat}\xspace}
\newcommand{\grant}{\mathsf{grant}\xspace}
\newcommand{\revoke}{\mathsf{revoke}\xspace}
\newcommand{\doManeuver}{\mathsf{doManoeuvre}\xspace}
\newcommand{\getReachableAgents}{\mathsf{getReachableAgents}\xspace}
\newcommand{\getARSet}{\mathsf{getARSet}\xspace}
\newcommand{\getPossibleTurns}{\mathsf{getPossibleTurns}\xspace} 
\newcommand{\storeMR}{\mathsf{storeMR}\xspace}

\newcommand{\startTimer}{\mathsf{startTimer}\xspace}
\newcommand{\stopTimer}{\mathsf{stopTimer}\xspace}
\newcommand{\tryManeuver}{\mathsf{tryManoeuvre}\xspace}
\newcommand{\aID}{\mathsf{aID}\xspace}
\newcommand{\clock}{\mathsf{clock}\xspace}
\newcommand{\position}{\mathsf{position}\xspace}
\newcommand{\velocity}{\mathsf{velocity}\xspace}
\newcommand{\acceleration}{\mathsf{acceleration}\xspace}
\newcommand{\storeAR}{\mathsf{storeAR}\xspace}
\newcommand{\getMR}{\mathsf{getMR}\xspace}
\newcommand{\last}{\mathsf{last}\xspace}
\newcommand{\precedes}{\mathsf{myPriority}\xspace}

\newcommand{\getUnsafeAgents}{\mathsf{getHigherPriorityAgents}\xspace}

\newcommand{\ntfEst}{\mathsf{ntfyEstimat}\xspace}
\newcommand{\SM}{\mathcal{SM}\xspace}
\newcommand{\M}{\mathcal{M}\xspace}
\newcommand{\D}{\mathcal{D}\xspace}
\newcommand{\U}{\mathcal{U}\xspace}
\newcommand{\R}{\mathcal{R}\xspace}

\newcommand{\tRETRY}{\mathsf{tRETRY}\xspace}
\newcommand{\getSegmentData}{\mathsf{getSegmentData}\xspace}
\newcommand{\sID}{\mathsf{sID}\xspace}
\newcommand{\TdoForeverLoop}{{T_A}\xspace}

\newcommand{\messageForm}[1]{\langle #1 \rangle}

\usepackage{tikz}

\tikzstyle{state}=[shape=circle,draw=blue!50,fill=blue!20]
\tikzstyle{observation}=[shape=rectangle,draw=orange!50,fill=orange!20]
\tikzstyle{lightedge}=[<-,dotted]
\tikzstyle{mainstate}=[state,thick]
\tikzstyle{mainedge}=[<-,thick]

\SetKwInOut{Input}{input}
\SetKwInOut{Output}{output}

\newenvironment{proof}{\noindent{\textbf{ Proof.}}}{\hfill$\blacksquare$}

\newcommand{\sre}{$\mathcal{S}_{RE}$\xspace}
\newcommand{\smn}{$\mathcal{S}_{MN}$\xspace}
\newcommand{\sremn}{$\mathcal{S}_{RE+MN}$\xspace}

\newcommand{\tcno}{$\mathcal{TC}_{Normal}$\xspace}
\newcommand{\tcnno}{$\mathcal{TC}_{Noise}$\xspace}
\newcommand{\tccno}{$\mathcal{TC}_{Comloss}$\xspace}
\newcommand{\tcono}{$\mathcal{TC}_{Offender}$\xspace}
\newcommand{\tcre}{$\mathcal{ES}\,_{Normal}^{RE}$\xspace}
\newcommand{\tcnre}{$\mathcal{ES}\,_{Noise}^{RE}$\xspace}
\newcommand{\tccre}{$\mathcal{ES}\,_{Comloss}^{RE}$\xspace}

\newcommand{\tcomn}{$\mathcal{ES}\,_{Offender}^{MN}$\xspace}
\newcommand{\tcremn}{$\mathcal{ES}\,_{Normal}^{RE+MN}$\xspace}
\newcommand{\tcnremn}{$\mathcal{ES}\,_{Noise}^{RE+MN}$\xspace}
\newcommand{\tccremn}{$\mathcal{ES}\,_{Comloss}^{RE+MN}$\xspace}
\newcommand{\tcoremn}{$\mathcal{ES}\,_{Offender}^{RE+MN}$\xspace}

\newcommand{\vhigh}{$V_{H}$\xspace}

\newcommand{\vlow}{$V_{L}$\xspace}

\newcommand{\vhighttg}{$TTG\,\mbox{\footnotesize $(V_{H})$}$\xspace}

\newcommand{\vlowttg}{$TTG\,\mbox{\footnotesize $(V_{L})$}$\xspace}

\newcommand{\TTG}{$TTG$\xspace}

\newcommand{\chapter}{\section}
\newcommand{\Section}{\subsection}
\newcommand{\Subsection}{\subsubsection}
\usepackage{fullpage}
\begin{document} 


\title{Membership-based Manoeuvre Negotiation\\ in Autonomous and Safety-critical Vehicular Systems\footnote{This technical report is directly based on a master thesis~\cite{mastersthesis} written by Emelie Ekenstedt. The supervisor was Elad M.\ Schiller.}\\(preliminary report)}
\author{Ant\'onio Casimiro~\footnote{Faculdade de Ci\^encias, Universidade de Lisboa, Lisboa 1749-016, Portugal. E-mail: \texttt{casim@ciencias.ulisboa.pt}} \and Emelie Ekenstedt~\footnote{Department of Engineering and Computer Science, Chalmers University of Technology, Gothenburg, SE-412 96, Sweden, \texttt{emeeke@student.chalmers.se}} \and Elad M.\ Schiller~\footnote{Department of Engineering and Computer Science, Chalmers University of Technology, Gothenburg, SE-412 96, Sweden, \texttt{elad@chalmers.se}.}}
\maketitle

\begin{abstract}
A fault-tolerant negotiation-based intersection crossing protocol is presented. Rigorous analytic proofs are used for demonstrating the correctness and fault-tolerance properties. Experimental results validate the correctness proof via detailed computer simulations and provide a preliminary evaluation of the system performances. The results are compared to the ones that can be achieved via a risk estimator with and without combining the proposed protocol. Our fault model considers packet-loss, noisy sensory information and malicious driving. Our preliminary results show a reduction in the number of dangerous situations and vehicle collisions.   
\end{abstract}





\chapter{Introduction}

Autonomous vehicles need a way of assessing risk in order to keep their passengers and other road users safe. Without human inputs, the vehicle has to rely on information from its sensors, and communication with other vehicles and infrastructures. This information must then be interpreted and processed to form a view of the current traffic situation and estimate the likelihood of a future collision, so that the vehicle can take action if the risk gets too high.

Another method of reducing collisions is to let the vehicles follow a protocol that allows them to negotiate their way through traffic situations where lanes either run in parallel or intersect -- i.e. where \textit{critical sections} exist. One such protocol for \textit{manoeuvre negotiation} was developed by Casimiro and Schiller~\cite{ManoeuvreNegotiation} and an extended version of this protocol is presented in this thesis. The protocol defines a priority system where vehicles have to request permission from higher priority vehicles to enter a critical section. Granting a request implies giving the requester higher priority. Therefore, the negotiation process provides both a way for the requester to inform other vehicles about its intended manoeuvre and for the requester to ensure that the other vehicles will stop for it if something goes wrong.

Instead of only focusing on collision avoidance, the protocol also tries to prevent \emph{priority violations}. A \textbf{priority violation} occurs if a vehicle has to either slow down or stop for another vehicle that is considered to have lower priority according to the traffic rules. The protocol is also designed to keep vehicles safe in the event of network failures, such as when messages are delayed or lost.

\Section{Problem formulation}




Collision mitigation for autonomous vehicles is an inherently difficult task. The risk assessments needed for mitigating collisions have to be based on information from sensors where noise is always present. The sensors' accuracy may further be reduced by certain weather conditions and objects blocking their view. Allowing vehicles to broadcast their own position to other vehicles may increase the accuracy of the information available for assessing risks. However, for this setup to have a reliable source of information, it is required that no sensor is faulty and that no vehicle has malicious intentions.

The task of collision mitigation can also be performed by making vehicles reveal their intentions instead of trying to infer intentions from positional data. Inferring intentions is a classification problem and it is hard to correctly classify the intention of a vehicle that engages in unusual driving behaviour. Alternately, if vehicle intentions are disclosed via messages, then the system can be exploited by malicious drivers. 

To further reduce the number of collisions it is also important that vehicles can negotiate with each other and not only communicate state information. Enabling negotiation implies that the vehicles can collaborate to plan safe trajectories through the sections of road that they will share. Efficiency and fault tolerance can be increased if the negotiating process follows a protocol. To ensure safety, the protocol must be carefully designed and analysed, and undergo rigorous testing. Since the system requires that vehicles follow the protocol to mitigate collisions, this setup is also vulnerable to abuse by malicious drivers.

All approaches to collision mitigation have weak points. Therefore, offsetting these weaknesses by combining different systems with complimentary strengths is critically important if we wish to create a reliable safety system for autonomous vehicles.

\Section{Purpose} 

As vehicles with autonomous abilities become more and more common, the need for more robust safety systems increases. A lot of research have been made -- and still is being made -- in the area of predicting vehicle motion and risk estimation. These predictions and estimations can be used to adjust the ego vehicles trajectory in order to avoid collisions. Another approach to implementing a safety system is use communication, vehicle-to-vehicle or vehicle-to-infrastructure, with a protocol determining how messages should be sent. Less time have been spent in this field and in how the two types of safety systems can be combined to create a system less sensitive to noise and network failures etc. The purpose of this thesis is to contribute to the latter field by providing an analytical correctness proof of a safety system based on a communication protocol and by presenting results from running simulations with this safety system, a safety system based on risk estimation, and the two systems combined.

\Section{Delimitations}

The implementation of our safety system will only be tested in computer simulations and not on real vehicles. Due to time limitations, the simulations in this project focus only on the highly risky manoeuvre of left turn across the priority lane in an intersection. Another limitation considers the amount of vehicles that are allowed to be in the intersection at the sane time, which is reflected by the number of concurrently granted vehicles. The version of the communication protocol that is presented in this thesis only allows one held grant per vehicle. 

\Section{Disposition} 

A short literary review of various methods for performing risk estimation for autonomous vehicles is found in Chapter \ref{sec:background}. Chapter \ref{sec:Solution} presents the key concepts of the extension of the manoeuvre negotiation algorithm and is followed by a detailed algorithm description and correctness proof in Chapter \ref{sec:AlgDescrProof}. The software architecture for performing simulations are described in Chapter \ref{sec:implementation} and the evaluation setup and criteria is presented in Chapter \ref{sec:evaluation}. Results from the simulations can be found in Chapter \ref{sec:results} which is followed by a discussion of the results and suggested further research in Chapter \ref{sec:discussion}. Chapter \ref{sec:conclusion} finally presents the main conclusions of this thesis.

\Section{Contribution} 

Risk estimation for autonomous vehicles is currently a highly researched topic. A lot of effort has been put into develop individual safety systems for particular traffic situations that could meet the strict regulations on safety for autonomous vehicles. Combining individual safety systems to achieve higher safety has not received much attention yet and it is this area that this thesis will make a contribution to. 
This thesis presents a way to combine vehicles' individual risk estimation with a distributed, fault tolerant manoeuvre negotiation protocol to increase safety in intersections. The high-level abstraction of the protocol does not limit its usage to intersections and it could potentially be used in many different traffic situations.
This thesis also presents an analytic proof of correctness for the protocol. Results from simulations where communication loss and extra noise are added further give an insight into how well the combined safety systems preform in non-ideal situations.


\chapter{Related work}
\label{sec:background}






Risk Estimation involves three core components: inferring the behaviour of other vehicles, defining a useful risk metric, and computing the risk based on the inferred behaviour. The aforementioned behaviour may be a composition of the manoeuvre that the vehicle is likely to perform and whether or not it is following traffic rules and conventions. 

Detecting high-risk situations in time is vital for avoiding collisions. Another approach to collision mitigation is to design V2V-communication protocols that prevent risky situations from occurring. Here follows a summary of related work in the area of \emph{risk estimation} and \emph{collision mitigation}.

\Section{Behaviour Inference}

Machine learning approaches have been shown to function well as classifiers for driver behaviour. A \textbf{Recurrent Neural Network} (RNN) with a \textbf{Long Short Time Memory} (LSTM) was used by~\cite{LongShort} to infer manoeuvre intention at intersections. RNNs have a feedback loop which gives them a form of memory and LSTM cells can remember a value for an arbitrary amount of time, which is useful when predicting manoeuvres.

A RNN classifier was also used to predict driver intention in roundabouts using short segments of LiDAR tracking data~\cite{RNN}. 




A comparison of different classification methods for \textit{signalled} intersections was performed by \cite{IntersectionTTIRDPSDR}. The authors developed two classifiers: one based on a \emph{Support Vector Machine} (SVM) combined with a Bayesian filter and one based on \emph{Hidden Markov Models} (HMMs), having trajectory parameters such as vehicle speed, lateral position etc. 

Three classification algorithms based on static time to intersection, the acceleration needed to stop at the stop line, and if the vehicle has a high speed in relation to it's distance from the intersection -- compared to a fitted regression curve fit to normal stop-patterns -- were also evaluated. The authors found that the SVM performed much better than all the other classification methods, with the HMM in second place. The authors reasoned that the data sets that the HMM used to attempt to generalise driving patterns contained large outliers which strongly affected its results. The SVM, on the other hand, tries to find a \textit{separating boundary} between the sets and thus is less affected by these same outliers.

An \textbf{Intention Aware} method was developed by~\cite{Lefevre2013IntentionAwareRE} to detect unexpected behaviour in \textit{give-way} intersections. Their method uses a HMM with predefined transition probabilities and a speed model of constant speed. The behaviour of lower priority vehicles is modelled as following or not following the convention of giving way to higher priority vehicles\footnote{see Section \ref{sec:RE} for more details}. 

Other aspects of intention aware motion prediction involves the behaviour to follow other vehicles. An intelligent driver model was used by~\cite{PrecedingVehicle} to account for that vehicles usually adapt to the speed of a vehicle closely preceding it.

\Section{Risk Metrics}

In this context, the commonly applied definition of risk is \textbf{Time to Collision} (TTC), which refers to the time remaining until a collision will occur if no object changes its intended trajectory. 
The authors of \cite{MSM_lanechange} relied on both TTC and another risk metric called \textbf{Minimal Safety Margin} (MSM) -- which provides a measure of the distance to other vehicles -- and managed to detect all risky situations involving lane change on highways in their tests. Their method involved using a Bayesian network to both account for uncertainties and estimate risk for each of the different lanes. They combined this with time-window filtering to produce a robust risk assessment method.

The concept of \textbf{Looming} was introduced by \cite{Looming} to compute risk for \textit{general traffic situations} -- i.e. circumstances where predefined trajectories may not be available -- that would perform better than mere TTC. Looming measures how much space an object occupies in the \textit{visual field} of a vehicle. A positive loom rate signifies that the object is approaching the vehicle and vice versa. The authors used a \textbf{Support Vector Machine} (SVM) trained using a combination of TTC and loom data to detect collisions. 


Approaches have been made to not only detect collisions, but also distinguish more severe collisions from minor ones. For this purpose, G. Xi et al~\cite{Situational} presented a situational assessment method that includes the system's internal energy in the risk computation. Internal energy is computed with reference to the vehicle's reduced mass and relative velocity, and provides an indication of the severity of a potential impact.


An interaction-aware model similar to the one in \cite{Lefevre2013IntentionAwareRE} is combined with a \textbf{Rapidly-exploring Random Tree }(RRT) classifier by \cite{RRT-and-French} in an attempt to increase safety by paying more attention to the traffic environment. The classifier interprets a traffic situation as either ''dangerous'' or ''safe'', enabling a vehicle to be more cautious in situations that are more prone to collisions.

 Many Risk Estimation methods are designed for particularly defined traffic situations -- such as straight highways or 4-way intersections --  but in reality, an autonomous vehicle must be able to handle many different situations. A scenario-adaptive system was developed by Geng et al.~ \cite{Scenario_adaptive} to infer driver behaviour of vehicles in more than one traffic situation. A HMM was trained using collected data for a number of differing behaviour classes, with each class corresponding to distinct manoeuvres (e.g. left turn at an intersection, lane keeping etc.). With the introduction of \textit{priori} information (i.e. traffic rules and conventions) and the use of Bayesian inference, \cite{Scenario_adaptive} also showed that the behaviour prediction time horizon could be increased by on average  $56\%$ for lane change and $26\%$ for long time precision. Scenarios were modelled with an ontology model so that scenarios could be inferred using sensory information -- to identify features in the environment -- and rule-based reasoning.
 
  A model-based algorithm calculating the acceleration or steering angle needed to avoid a collision was presented by \cite{AvoidCollision}. An emergency break or lane-change manoeuvre could be triggered if the output values rise over a set threshold in order to avoid collisions. 

 \Section{Traffic Coordination Algorithms}

 Autonomous vehicles can make use of V2V-communication to negotiate and reach an agreement over who was the right to enter critical sections and thereby mitigate collisions. An approach using a \textbf{temporarily elected leader} was presented by \cite{VirtualTrafficLight} to direct traffic in intersections. The elected leader decides which direction will have right-of-way, analogous to a green light, until a new leader is elected. The approach introduces a way to control traffic flow through both signalled and unsignalled intersections and was shown to reduce traffic accidents by $70\,\%$ in simulations.
 
 Algorithms have also been developed to find an optimal of coordinating traffic through intersections. A \textbf{Sequential Quadratic Programming }(SQP) algorithm was developed by \cite{OptimalCoordination} for this purpose. The formulated \textit{distributed optimisation} problem considers the times when each vehicle enters and exits the intersection, as opposed to considering the complete trajectories, in order to reduce the amount of communicated data between vehicles.
 

 \Section{Cooperative Approaches}
 
 Detecting \textit{temporarily occluded objects} is inherently hard since most  (often light-based) sensors used in vehicles cannot penetrate opaque materials. Cooperation between vehicles could potentially increase safety, since the object that blocks the sensor's view may actually be another vehicle. Such an approach was developed by \cite{CooperativePerception}. This approach involved vehicles using V2V-communication to inform other vehicles of their own state as well as information about detected -- possibly hidden -- objects.
 



\chapter{Solution} 
\label{sec:Solution}

This section presents some of the main ideas that have been developed during this project and serves as an introduction to the Manoeuvre Negotiation Protocol presented in Chapter \ref{sec:AlgDescrProof}. The following sections will cover the main changes made to the original Manoeuvre Negotiation Protocol, introduce the concept of priorities, describe the granting process, and explain how risk estimation is used in our system.

\Section{Adding an explicit release of a received grant}
\label{sec:explicit}


The original protocol for manoeuvre negotiation was designed from a scheduling point of view with focus on both safety and throughput. A permission to change the traffic priorities was bound to a timer set to a predefined manoeuvre time and the permission expired when the timer expired. This approach ensures that no given permission last forever, thus preventing deadlocks, but it assumes that the granted vehicle can finish its manoeuvre within the given time frame. Consider, for example, a granted vehicle $p_i$ that has to stop in an intersection and another vehicle $p_j$ whose grant timer just expired. The protocol would then allow vehicle $p_j$ to be granted access to the intersection, resulting in a collision unless another safety system successfully detects the dangerous situation in time.

\Subsection{From scheduling to setting priorities}
\label{sec:scheduleToPrio}
Instead of viewing the intersection as a scheduling problem with equal focus on throughput and safety we wish to put more focus on safety by introducing negotiable priorities not bound to timers. The first step towards a priority oriented protocol is to remove the grant timer and to introduce an explicit $\messageForm{RELEASE}$ message to end a grant instead. A granted vehicle will send this release message after it has left the intersection to notify vehicles holding a grant for it to release the grant. This approach assures that a grant will not expire too early and is able to ensure safety even if the granted vehicle slows to a stop in the intersection. Explained from the new priority perspective, a granted vehicle is given a higher priority and none of the granting vehicles, with lower priority, are allowed to enter the intersection until assured that the granted vehicle has left the intersection and they themselves have received a higher priority. For more details on priorities\footnote{see Section \ref{sec:Priorities}}.

The original protocol used timers to prevent deadlocks in the presence of message loss. Removing the grant timer and introducing a new message type implies that an additional feature has to be introduced in order to keep preventing deadlocks. The solution we propose is to let the granting vehicles regularly calculate if the granted vehicle has left the critical section. In this project the calculations are based on positional and road map data sent or shared between vehicles, but it would also be possible to use LiDAR, cameras, and other perception tools for autonomous vehicles for this purpose. A granting vehicle that calculates that the granted vehicle has left the intersection will release the grant in the same way as if it had received a $\messageForm{RELEASE}$ message. A lost $\messageForm{RELEASE}$ message will thereby not cause a deadlock.

\Section{Priorities}
\label{sec:Priorities}

The concept of priorities is used by vehicles to determine if they have the right to perform certain manoeuvres when other vehicles are present. More specifically, priorities determine who has the right of way when the planned trajectories of two vehicles intersect. 





Priorities can be determined by traffic rules or be indicated by, for example, traffic lights. The former type of priorities do not change over time and we refer to them as \textit{static} priorities. The latter priorities are time dependent and we refer to them as \textit{dynamic} priorities. 

Here, we define priorities as the order in which vehicles will be allowed to enter the intersection. The priority order depends on which lane the vehicles are in and how far in space and time they are from the intersection. Our system is based on dynamic priorities where vehicles' priorities are negotiated through our manoeuvre negotiation protocol. Static priorities determined by traffic rules are used as a base state for our dynamic priority system. This section further explains the concept of priorities and how it is used in the Manoeuvre Negotiation Protocol.



\Subsection{Priority matrix}
\label{sec:prioMatrix}
The Manoeuvre Negotiation Protocol relies on manoeuvre specific memberships, which for a vehicle $p_i$ are all the vehicles with higher priority than $p_i$ when $p_i$ wishes to perform the manoeuvre bound to the membership. The default priorities in the negotiation protocol are decided by traffic rules. For an intersection, the traffic rules can be summarised in what we refer to as a \emph{priority matrix}.
Each lane leading to the intersection is assigned a priority matrix, where rows correspond to possible manoeuvres and columns represent the other directions into the intersection. Every cell takes on the value $\true$ or $\false$. A $\false$ implies that a vehicle coming from the direction associated with the matrix and intending to do the manoeuvre described by the row, must ask the agents coming from the lane corresponding to the column for permission to change the traffic priorities temporarily. An example of a priority matrix for vehicles on a straight priority road in a 4-way intersection can be found in Figure \ref{fig:PrioMatrix}. Vehicles in a lane on the priority road must ask for permission to temporarily change the traffic priorities from vehicles in the opposing lane only if they intend to do a left turn across path, indicated by the $\false$ at the top of the first column. The priority matrix for the opposing lane going in the opposite direction will in in this case be the same since the intersection is symmetrical. Vehicles approaching the intersection from a non-priority road will have priority matrix different from the one for vehicles on the priority road.

\begin{figure}
    \centering
\begin{minipage}[t]{0.40\textwidth}
\includegraphics[width=\textwidth]{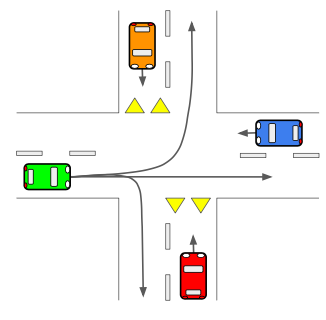}
\end{minipage}
\begin{minipage}[t]{0.50\textwidth}
\includegraphics[width=\textwidth]{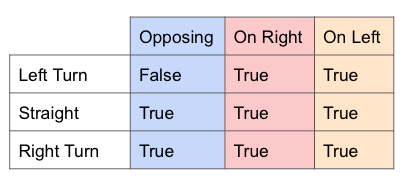}
\end{minipage}
\caption{A priority matrix for the green vehicle coming in from the left on the priority road. A $\false$ infers that the green vehicle has lower priority for the turn described by the row in relation to the lane described by the column. }
\label{fig:PrioMatrix}
\end{figure}

\Subsection{Priority levels}
All vehicles approaching an intersection can be put in a hierarchical order of priority. This order can be visualised as levels on a scale, where the top level determines the group of vehicles that currently have permission to enter the intersection, which are the vehicles with the highest priority. Vehicles in a group of lower priorities will have to wait for higher priority vehicles to exit the intersection before they can enter. This rule is just a direct translation from what the traffic priorities say but the introduction of this concept makes it easier to visualise and perhaps comprehend. The introduction of the Manoeuvre Negotiation Protocol enables vehicles to negotiate with the higher priority vehicles for a temporary rise in its priority and a safe passage through the intersection.

It is worth noting that there can be as many priority levels as there are vehicles on the road since the levels represent an observer's perspective. A vehicle's perspective, however, only needs a maximum of 3 levels corresponding to higher, equal or lower priority than the ego vehicle. This is because a vehicle only needs to compare its own priority to others to navigate through traffic since it just needs to know which vehicles to give way to.


 \begin{figure}
    \centering
\begin{minipage}[t]{0.40\textwidth}
\includegraphics[width=\textwidth]{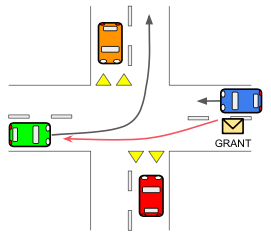}
\end{minipage}
\begin{minipage}[t]{0.40\textwidth}
\includegraphics[width=\textwidth]{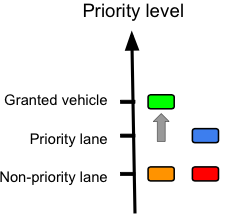}
\end{minipage}
\caption{The green vehicle coming from the left is granted to temporarily increase its priority level. The green vehicle is now allowed to enter the intersection }
\label{fig:Priority_grant}
\end{figure}

\Subsection{Priority violation}
\label{sec:prioViolation}
The two main ideas of the Manoeuvre Negotiation Protocol is firstly to let traffic flow in a way similar to when the protocol is not in use when there are no communication failures. Secondly, when communication failures are present, the protocol is meant to add an extra layer of safety by making the vehicles follow the negotiated dynamic priorities. If, for example, a vehicle has been granted to rise its priority but some circumstance causes the vehicle to stop in the intersection, then other vehicles will not be allowed to enter the intersection until they can confirm that the granted vehicle has left it. This may cause vehicles that initially had higher priority (as specified by traffic rules) to slow down or stop for initially lower priority vehicles. We refer to such events as \textit{priority violations}.

The first main idea can be reformulated to say that priority violations will be avoided when there are no communication failures. To achieve this, our negotiation protocol use the functions $\noPriorityViolationStat$ and $\noPriorityViolationDyn$ to determine if a manoeuvre can be performed without leading to a priority violation. The former takes normal and unusual driving patterns into account while the latter only considers normal driving patterns. More details of the two functions and their usage areas are provided in Section \ref{sec:unsafeAgents} (Section \ref{sec:nPV_Static_alt} for an alternative approach) and \ref{sec:nPV_Dynamic}.

\Subsection{Algorithm Details of \texorpdfstring{$\noPriorityViolationStat$}{Lg}}
\label{sec:unsafeAgents}

$\noPriorityViolationStat$ is used to calculate which vehicles $p_j$ that have higher priority than $p_i$ -- where priorities are defined by the priority matrix for the intersection-- and could possibly reach the intersection in the time $t_{max}$ it takes for $p_i$ to both request a priority change and perform its manoeuvre through the intersection. In order to account for unusual driving behaviour, we defined a threshold distance $d_{max}$ such that any vehicle starting from a distance $d>d_{max}$ away from the intersection will not be able to reach it within $t_{max}$ time units. 

The implementation of $\noPriorityViolationStat(p_i,p_j)$, for two vehicles $p_i$ and $p_j$, first determines if $p_i$ has higher priority than $p_j$ according to the priority matrix and if so returns $\true$. If $p_i$ has lower priority the procedure returns $\true$ if and only if $p_j$ is further than $d_{max}$ length units away from the intersection.


\Subsection{Algorithm details of \texorpdfstring{$\noPriorityViolationDyn$}{Lg}}
\label{sec:nPV_Dynamic}

The implementation approach of $\noPriorityViolationDyn$ makes an assumption about driver behaviour in order to try to increase throughput. This approach requires predefined trajectories and speed models corresponding to the manoeuvres going straight, turning left or right, and stopping. 

Let $p_i$ and $p_j$ be two vehicles approaching an intersection. A call to \\ $\noPriorityViolationDyn(p_i, turn_i, p_j)$ will return $\true$ if the vehicles are estimated to never be in the intersection at the same time. The estimation of time period spent in the intersection -- here called the \emph{occupation interval} -- is performed for both vehicles but the process is here described for $p_j$ only to avoid duplication. Normal driving pattern for $p_j$'s current intended manoeuvre is used to estimate the time it will take $p_j$ to reach the intersection ($TTI_j$) and the time it will take for it to leave the intersection ($TTE_j$). The \textbf{occupancy interval} for $p_j$ is then defined as $O_j = [TTI_j,TTE_j]$. The intention of $p_j$ may be unknown, in which case $O_j$ is a set consisting of the predicted intervals for each of the four driving manoeuvres. 
An additional margin of $\chi(TTI_j-T)$ -- where $T$ is the \textit{current time} and $\chi$ is a constant -- is subtracted from each $TTI_j$ and a margin of $\chi(T-TTE_j)$ is added to each $TTE_j$ to help avoid near collisions. The constant $\chi$ affects the uncertainty in the estimation of $O_j$. Thus, more uncertainty is added to $O_j$ for an agent $p_j$ that is further away in time from the intersection than an agent that is closer. This increase in uncertainty accounts for some deviations from the defined normal driving pattern. We define $\mathcal{O}_j$ as the \textbf{increased occupancy interval} for an agent $p_j$.

In the next step, the increased occupancy intervals for $p_i$ and $p_j$, $\mathcal{O}_i$ and $\mathcal{O}_j$, are compared to evaluate the traffic situation. No priority violation is estimated to occur between $p_i$ and $p_j$ if $\mathcal{O}_i \cap \mathcal{O}_j = \emptyset$. 

This approach cannot guarantee that no priority violation -- and thereby that no collisions -- will occur if vehicles exhibit behaviour that deviates extremely from the normal driving patterns and priorities are kept static. However, enabling priorities to be dynamic by introducing negotiation according to the \textit{Manoeuvre Negotiation protocol} ensures that no collision will occur as long as all vehicles follow the protocol. This implementation of $\noPriorityViolationDyn$ is therefore suitable for use during the negotiation process when agents make decisions on granting or denying another agent. The approach provides higher throughput than the approaches presented in Section \ref{sec:unsafeAgents} and \ref{sec:nPV_Static_alt}. 

\begin{algorithm}[ht!]
	\begin{\VCalgSize}
		\textbf{Structures:}\\
		$AgentState = (ta,x,v,a)$\tcp*[r]{timestamp, position, velocity, acceleration}
		\ \\
		\textbf{Interfaces:}\\
	    $getTimeToIntersection(p[, turn])$: \;
	    $getTimeToExit(p, TTI[, turn])$: \;
		\ \\
		\textbf{procedure} $getOccupationInterval(p, turn = None)$
		\Begin{
		    \textbf{let} $TTI = getTimeToIntersection(p, turn)$\;
		    \textbf{let} $TTE = getTimeToLeave(p, TTI, turn)$\;
		    \textbf{let} $TTI = TTI - (TTI - p.ts) \chi$\;
		    \textbf{let} $TTE = TTE + (TTE - p.ts) \chi$\;
		    \textbf{return} $[TTI, TTE]$
		}
		\textbf{procedure} $\noPriorityViolationDyn(p_i, turn_i, p_j)$ \Begin{
		    \textbf{let} $TTI_i, TTE_i = getOccupationInterval(p_i, turn_i)$\;
		        \textbf{let} $TTI_j, TTE_j = getOccupationInterval(p_j)$\;
		        \lIf{$((TTI_i \geq TTI_j) \land (TTI_i \leq TTE_j)) \lor ((TTI_j \geq TTI_i) \land (TTI_j \leq TTE_i))$}{
		        \textbf{return} $\true$
		       }\lElse{\textbf{return} $\false$}
			\textbf{return} $unsafeAgents$ \;
		}
	\end{\VCalgSize}
	\caption{$\noPriorityViolationDyn$ relying on dynamic priorities for safety}
	\label{alg:nPVDynamic}
\end{algorithm}


%


\Section{Deciding to grant or deny}
\label{sec:interval}

In our system the traffic priorities at any point in time are determined by the static traffic rules and the dynamic changes of priorities established by the manoeuvre negotiation protocol. A vehicle that wants to perform a manoeuvre that would change the current traffic priorities will therefore have to ask for permission to do so from all vehicles in its membership. 
The choice of reply to a request, either $\messageForm{GRANT}$ or $\messageForm{DENY}$, depends on the output of the function $\noPriorityViolationDyn$ which computes \emph{occupation intervals} -- time intervals in which the requester and requestee are predicted to be inside the intersection.

\Subsection{Predicting intersection occupation intervals}

The prediction of our occupancy intervals used in $\noPriorityViolationDyn$ are based on normal driving patterns defined as a speed model, see Figure \ref{fig:SpeedModel}. The speed model defines the ideal speed based on a vehicle's intention and its distance from the intersection. Starting from the most recently known position at time $T$ of a vehicle $p_i$, the time to reach the intersection, $TTI_i$, and time to exit the intersection, $TTE_i$, is computed using the speed $v$ determined by the speed model and by integrating $1/v$ from the current position up until reaching the intersection. The uncertainty constant $\chi$, which reflects how well we believe that vehicles follow the speed model, was set to 0.25 here since the vehicles in our simulation do not deviate much from the speed model.

\begin{figure}[!ht]
    \centering
    \includegraphics[width=\textwidth/5*4]{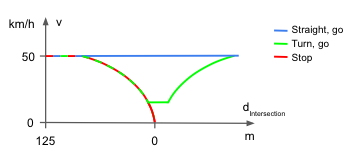}
    \caption{A speed model defining ideal speed v for a vehicle based on its distance from the intersection d$_{Intersection}$ and intention which can be either going straight, turning left or stopping. The model is based on alternating constant speed and constant acceleration.}
    \label{fig:SpeedModel}. 
\end{figure}
 
\Subsection{\texorpdfstring{$\noPriorityViolationDyn$}{Lg}'s role in the grant decision}

A Requestee calls the interface $\noPriorityViolationDyn$ presented in Algorithm \ref{alg:nPVDynamic} to determine if it's safe to grant the requester. The interface calculates the occupation intervals $\mathcal{O}_i$ and $\mathcal{O}_j$ for the requestee $p_i$ and the requester $p_j$ and calculates the gap $\Delta t_{gap}$ between them when put on the same time line, see Figure \ref{fig:intervals}. Any overlap in the intervals, $\Delta t_{gap}<0$, signifies a higher collision risk and will result in a returned $\false$ causing $p_i$ to reply with a $\messageForm{DENY}$. No overlap among the two time intervals results in a returned $\true$ which enables $p_i$ to continue checking the other state related criteria needed when making the grant decision\footnote{see Section \ref{sec:algDescrMNA}}.

\begin{figure}[!ht]
    \centering
    \includegraphics[width=\textwidth/5*4]{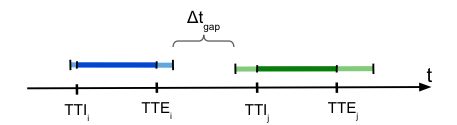}
    \caption{Two occupancy intervals for agent $p_i$ and agent $p_j$ with added safety margins. The time gap $\Delta t_{gap}$ is evaluated as the smallest gap between the two intervals outer margins. }
    \label{fig:intervals}
\end{figure}


\Section{Risk estimation interfaces}

The manoeuvre negotiation protocol needs risk estimation in order to provide safety, especially when deviations and failures are present in the system. Three main usage areas for risk estimation will be presented in this section where the risk estimation can be provided to the manoeuvre negotiation protocol in the form of interfaces. 

\Subsection{Decide when to grant or deny a request}

The receiver of a priority change request has to decide whether to reply with a $\messageForm{GRANT}$ or a $\messageForm{DENY}$. The decision has to be based on how much risk is involved with the proposed change of priority. In this case we define high risk as an overlap in the estimated time intervals of when the requester and requestee will be in the intersection\footnote{see Section \ref{sec:interval}}.



\Subsection{Decide when enough responses have been received}

The negotiation process relies on memberships to ensure that all agents necessary for ensuring safety are involved in a negotiation process.
The memberships are dynamic and are recalculated every $T_M$ time units. A membership retrieved just before sending a request may therefore change within the $2T_D$ time units that the requesting vehicle waits for replies. Some agents may be added and some removed from the membership. An agent that is removed during the waiting time may cause the requesting vehicle to wait unnecessarily, since an agent removed from the membership does not have higher priority anymore and their (possibly lost) reply is no longer relevant. A requesting agent should therefore filter the received replies to check if it has got replies from the intersection of the set of agents in the membership valid when the request was sent $\mathcal{D}$ and the current membership $MS$ to avoid this waiting period. More formally, the vehicle needs to check if $\mathcal{D} \cap M \cap \mathcal{R} = \phi$, where $\mathcal{R}$ is $\mathcal{D}$ minus the agents that have responded to the latest request and $\phi$ is the empty set.

\Subsection{Risk estimation as a background process}
\label{sec:RE_Background}
In the real world, outside simulations, one cannot assume that all vehicles will follow an ideal trajectory and follow the traffic rules at all times. An autonomous vehicle thus need to be able to detect a deviating and possible malicious vehicle and perform needed action in time to avoid a collision. Continuously running a risk estimation procedure, that can perform long-term predictions, as a background process increases the chances of detecting these deviating vehicles in time to avoid a collision. For this thesis we are using a slightly modified version of the risk estimation procedure presented in \cite{Lefevre2013IntentionAwareRE}\footnote{see Section \ref{sec:RiskEstimator}}. We refer to this process as the \emph{Risk Estimator}. When the Risk Estimator detects a risk over 75\% then it triggers an \emph{emergency break}\footnote{see Section \ref{sec:RiskEstimator} for its definition of risk}.


\subsubsection{Interconnected risk estimator and manoeuvre negotiation}

The change of priorities achieved with the manoeuvre negotiation protocol must be recognised by the risk estimator since the risk estimator uses priorities to calculate risk. For example, a vehicle $p_i$ on the priority road that have granted another vehicle $p_j$ should not detect a high risk when $p_j$ enters the intersection before $p_i$ since $p_i$ has agreed to give $p_j$ higher priority. If the risk estimator would not be made aware of the priority change then, depending on the time gap, it might detect a high risk since $p_j$ would be seen as a low-priority vehicle doing a turn in front of a high-priority vehicle.

The protocol calls the interface $\ntfEst$ defined in Figure \ref{alg:defs} with $\grant$ to tell the risk estimator about the priority change. The risk estimator puts the granted vehicle in a grant list identifying vehicles with higher priority, until the interface $\ntfEst$ is called again from the protocol with $\revoke$.



\chapter{The Manoeuvre Negotiation Algorithm}
\label{sec:AlgDescrProof}


\Section{System Settings}
\label{sec:systemSettings}
Traffic rules regulate entry allowances to avoid ambiguity and reduce collision risk when two or more vehicles arrive at an intersection's border. Each vehicle has to arrive at the intersection border in a way that allows the driver to observe all other road users that they must give way to. It is allowed to enter and cross the intersection before higher priority vehicles as long as it does not \emph{violate their priority} -- i.e. either cause a collision or force the higher priority vehicles to change speed, path or trigger their safety systems to perceive risk. 

We consider a road intersection with \emph{predefined priorities}. These priorities are determined by the \textit{give way} traffic rules, such that whenever two vehicles have intersecting paths, one of them has a predefined (strictly) higher priority that depends on the incoming lanes of the two vehicles. 

An \textbf{}{agent} is a \textit{computing and communicating entity} that assists in deciding whether a vehicle may cross the intersection. Such algorithmic agents can be components in a distributed automated driving system. Note that our model assumes that all vehicles are controlled by these agents and it does not focus on how to deal with malicious agents in a way that can represent (careless) human drivers, since we assume that all agents are attentive and follow the traffic rules. We also assume that the agents periodically transmit information about their whereabouts over communication channels that are prone to failures. 

Agents can perform manoeuvres concurrently, such that each agent performs a sequence of manoeuvres; one manoeuvre at a time. Agents have autonomous abilities that enable them to keep to their lanes and to not drive too close to the vehicles in front of them in the same lane. We assume that the agents have collision avoidance mechanisms. However, these mechanisms cannot facilitate (comfortable) intersection crossings since these mechanisms do not assess the risk of entering the intersection in a way that violates the traffic rules. For example, each agent must give way to higher priority traffic. This work focuses on agents' obligations to stop, if necessary, and let higher priority agents proceed. The exact decision as to whether a complete stop is needed depends on the traffic rules and situation. This can be determined by a separate algorithmic (non-distributed) component for risk estimation, such as the one by Lef{\`e}vre et al.~\cite{Lefevre2013IntentionAwareRE}, which in turn depends on manoeuvre prediction and other components. Our model assumes access to such a risk estimator and that the vehicle changes its speed (but not its path) whenever risk is perceived. 

\begin{definition}
	\label{def:priorityViolationOccurred}
	Suppose that agent $p_i$ must give way to agent $p_j$ since $p_i$ has lower priority than $p_j$. We say that a \emph{priority violation} has occurred if, and only if, $p_i$ decides to enter and cross the intersection before $p_j$ in a way that causes $p_j$'s risk estimator to observe the need to reduce speed.
\end{definition}


We note that in the presence of aberrant behaviour -- such as that caused by vehicular malfunction or malicious intent -- it is impossible to guarantee that no priority violation (Definition~\ref{def:priorityViolationOccurred}) will occur. For example, a vehicle may enter an empty intersection without any approaching vehicle in sight. Any aberrant event may then cause the vehicle to come to a permanent halt within the intersection, thereby blocking the way for other vehicles. For the sake of making the scope manageable, Assumption~\ref{req:noSafetyViolation} states that such events do not occur. Our model assumes that agent $p_k$ has access to the function $\noPriorityViolationDyn(@parameters)$, where $@parameters$ refers to timed information about the computational context, such that Assumption~\ref{req:noSafetyViolation} is respected.
The \textit{time aspects} defines a time horizon and concerns the time it will take a requesting agent to send and receive answers to a request and how long it will take to perform the requested manoeuvre. The function $\noPriorityViolationDyn$ must then evaluate the situation up until this time horizon.

We have considered two different implementations of $\noPriorityViolationDyn$ that provide different accuracy to throughput ratios. The first solution is based on patterns for normal driving in approaching an intersection. Speed models are used to estimate when vehicles will enter and exit the intersection. Vehicles deviating from the normal driving pattern are accounted for by adding an extra safety margin based on how far into the future the entry and exit occurs. Overlaps in the estimated intervals implies a potential priority violation. This approach is elaborated in Section \ref{sec:nPV_Dynamic}.

The second approach is based on occupancy prediction~\cite{Occupancy} and should be able to identify more situations in which priority violations can occur. However, the accuracy of predicting a future priority violation is favoured over keeping the throughput through the intersection high. Further details are given in Section \ref{sec:nPV_Static_alt}.

\begin{assumption} 
	\label{req:noSafetyViolation}
	Suppose that the \emph{requester} agent $p_i$ requests a permission at a time $t_i$ from all relevant \emph{requestee} agents that have a higher priority. Moreover, $p_j$ receives $p_i$'s request at time $t_j$, asserts that $\noPriorityViolationDyn(AR,turn)=\true$ and then replies to $p_i$ with a $\messageForm{GRANT}$, where $t_j$ is between $t_i$ and $t_i+ T_D$, and $T_D$ is a bound on the communication delay of a message sent from one agent to another. Furthermore, by the time $t'_i$, agent $p_i$ has received $\messageForm{GRANT}$ replies from all the relevant agents, where $t'_i$ is a time between $t_i$ and $t_i+2T_D$. In this case, we assume that the fact that agent $p_i$ starts crossing the intersection immediately after time $t'_i$, implies that it is safe to cross and no priority violation will occur.
\end{assumption}

We note that some benign violations of Assumption~\ref{req:noSafetyViolation} are not deemed to violate system safety as a whole. Specifically, consider the case where agent $p_i$, before time $t'_i$, stays outside the intersection and immediately after time $t'_i$ starts crossing the intersection and halts (unexpectedly and without a notification) before completing the manoeuvre. The use of a risk estimator, such as the one by Lef{\`e}vre et al.~\cite{Lefevre2013IntentionAwareRE}, allows ample time for the arriving agent $p_j$ to safety avoid a crash with $p_i$. To the end of accommodating the ability of avoiding such benign violations of Assumption~\ref{req:noSafetyViolation}, we assume that the algorithm can notify its (local) risk estimator module (by calling $\ntfyEstimat(state): state \in \{\grant,\revoke\}$) about temporary changes to the priority setting. E.g., if $p_j$ grants $p_i$ a permission to cross the intersection before itself then $p_j$ calls $\ntfyEstimat(\grant)$, and once $p_j$ learns about $p_i$'s exit from the intersection then $p_j$ calls $\ntfyEstimat(\revoke)$. In detail, Lef{\`e}vre et al.~\cite{Lefevre2013IntentionAwareRE} refer to behaviour that is expected from the agent $p_j$ when it has a priority over $p_i$. We simply assume that this priority information is an input to the risk estimator module that can change during the period in which $p_i$ is granted a permission to cross the intersection before $p_j$. We note that without this assumption, a risk can be falsely detected if $p_i$ starts to cross the intersection before $p_j$.

\Subsection{Synchrony in the presence of communication failures}
All agents have access to a universal clock. We assume that the system uses mechanisms for omitting messages that do not arrive in a timely manner. That is, a timely message is delivered at most $T_D$ time units after being sent. Non-timely messages are not processed by the algorithm. Let $E$ be a system execution in which any sent message is delivered in a timely manner, i.e., no message is lost or omitted, and all delivered messages are sent during $E$. In this case, we say that $E$ is a \emph{synchronous} execution. All other executions are simply referred to as \emph{non-synchronous}. We assume that during non-synchronous executions, the \emph{communication is fair}. That is, any message that is sent infinitely often arrives infinitely often in a timely manner. However, we do not bound the time between two consecutive timed message arrivals. We note that the latter implies unbounded, yet finite, message delays. Since we consider the general context of wireless communication, we assume that the system without notice can reach a non-synchronous execution for any unbounded (yet finite) period. 

\Subsection{Task requirements}
The correctness proof will demonstrate both safety (Requirements~\ref{req:MembershipService}, \ref{req:RequestInitialisation}, \ref{req:GrantingDenying}, \ref{req:explicitRelease}, and \ref{req:AllowingDeferring}) and progression (Requirement~\ref{req:Progression}) for synchronous executions. For the case of non-synchronous executions, the proof considers the demonstration of safety (without considering progression).

\Subsection{Limitations}
For the sake of simple presentation, we make the following assumptions. We assume that an agent can request only one manoeuvre at a time since we consider the task of multiple manoeuvre requests to be related to planning. We consider a system that at any time allows at most one grant to a single agent; systems that consider concurrent grants are relevant extensions to the proposed solution. Once an agent is granted a permission to perform a manoeuvre, we require the system to allow the agent to either (i) finish this manoeuvre within $T_{Man}$ time units or (ii) to communicate its position in a timely manner to all agents that might collide with it (and by that notify them about the delay in the termination of the manoeuvre). We note that a possible extension to the proposed solution exists in which the system can suspect any granted agent that does not fulfil (i) and (ii) to be faulty. Such extensions can, for example, trigger a recovery procedure that makes sure that the road (intersection) is clear before revoking the grant and resuming normal operation.

\Subsection{Cloud-based membership service}
Requirement~\ref{req:MembershipService} considers a Cloud-based service that computes new memberships every $T_M$ time units and performs the computation on state data sent periodically from all agents to a storage service. A \emph{membership} for an agent $p_i$ contains a set of identifiers of agents $p_j$ whose: (i) path can interest with the one of $p_i$'s paths, and whose (ii) lane has equal or higher priority than $p_i$'s lane, and whose (iii) distance to $p_i$ is at most $d_{max}$ length units. Since the priorities are determined by direction specific traffic rules, we assume that the membership information provides one membership per possible manoeuvre through the intersection.

The Cloud-based service associates a timestamp $ts$ with every instance $M$ of a membership. The timestamp of the oldest agent state data of the agents in $M$ determines $ts$. Moreover, we say that $M$ is \emph{fresh} if $T < ts + 2T_M$, where $T$ is the current time. When agent $p_j$ appears in $p_i$'s fresh membership $M$, we say that $p_j$ may cause a  \emph{priority violation} with $p_i$.

Every membership $M$ of an agent $p_i$ is also associated with a boolean value $MO$, which we call the \emph{Manoeuvre Opportunity} indicator. The membership service sets $MO = \true$ to indicate that all agents in the membership set are within the maximum communication range of $p_i$ or that the provided membership for some other reason should not be considered valid. Note that the latter assumption allows the Cloud-based service, when needed, to indicate that the functionality offered by the proposed solutions is temporarily unavailable due to some implementation constraints. A concrete implementation of the proposed solution might, for example, require the assumption of an upper bound, $N_{max}$, on the number of vehicles that can occupy the considered traffic scenario due to limitations related to space, and communication range. In addition to violations of $N_{max}$, other constraints could consider the occurrence of failures or limitations on velocities in which the offered functionality cannot be considered feasible and safe.

\begin{requirement}[Membership Service] 
\label{req:MembershipService}
Let $M$ be a membership for an \\agent $p_i$ that has a manoeuvre opportunity at a time $T$ and $M$ is fresh: $T < ts + 2T_M$. We require that $M$ contains all agents that could potentially lead to a priority violation with $p_i$ from $T$ to $T + 2T_D + T_{Man}$.
\end{requirement}

\Subsection{Requester-to-requestee dialogue} 
\label{sec:ProofRequirements}
We consider a pair of distributed state-machines that the requester and requestees execute; cf. agents $p_i$, and respectively, $p_j$ in Figure~\ref{fig:agentDialog}. We require that these agents use $\noPriorityViolationDyn()$ in a way that allows using Assumption~\ref{req:noSafetyViolation} safely and repeatedly. To that end, there is a need to take into  consideration both $\messageForm{DENY}$ messages and temporary timing failures due to the occurrence of non-synchronous executions. Therefore, the requester automaton includes a pair of states for facilitating a persistent transmission of request messages $\messageForm{GET}$. Thus, whenever it is time to send a request, the requester enters the state $GET$ directly or passes through state $TRYGET$ first before reaching state $GET$. From state $GET$ $p_i$ repeatedly multicast to all relevant requestees the message $\messageForm{GET}$. Moreover, once requester $p_i$ recognise that it has received a $\messageForm{DENY}$ message, $p_i$ defers communication by entering the state $TRYGET$ and stays there until it is time to return to $GET$. We note that all agents play both roles of requesters and requestees; the proposed solution does not distinguish between these roles. For the sake of simple presentation, the requirements presented here mostly follow the above role separation with one exception; Requirement~\ref{req:GrantingDenying} refers explicitly to the case in which a requester receives a request and thereby also has to play the role of a requestee.


\begin{requirement}[Request Initialisation] 
\label{req:RequestInitialisation}
Consider an agent $p_i$ that aims to perform a manoeuvre through an intersection and a membership $M$ that corresponds to the direction of that manoeuvre. We require $p_i$ to first request a permission from all agents in $M$ according to the required dialog presented in Assumption~\ref{req:noSafetyViolation}. Moreover, we allow $p_i$ to initiate such a request if (i) $M$ is fresh, and (ii) $MO = \true$. Otherwise, we require $p_i$ to check conditions (i) and (ii) within $T_A$ time units.
\end{requirement}

Requirement~\ref{req:GrantingDenying} uses the term \emph{status} when referring to the automaton state and assumes that the state of an agent can encode the identifier of the agent to which it has granted a permission. It also assumes the existence the function $\precedes(m)$ that can arbitrate between two concurrent requests. This function returns $\true$ if the current request of the calling agent has higher priority than the request encoded by message $m$. The implementation of $\precedes()$ can depend on the directions of the relevant agents, the time in which their requests were initiated as well as the agent identifiers. 

\begin{requirement}[Granting and Denying] 
	\label{req:GrantingDenying}
	Suppose that the requester agent $p_i$ multicasts a permission request, message $m=\messageForm{GET}$, at a time $t_i$, to all requestee agents in $\D = M$, where $M$ is $p_i$'s fresh membership at the time $t_i$ for which there is a positive indication for a manoeuvre opportunity. Moreover, $p_j$ receives $p_i$'s request at a time $t_j$, where $t_j$ is between $t_i$ and $t_i+ T_D$. We require agent $p_j$ to call the function $\ntfyEstimat(\grant)$ and reply with a $\messageForm{GRANT}$ to $p_i$ if, and only if, $\noPriorityViolationDyn()=\true$ and: (i) $p_j$'s status is $NORMAL$ or $TRYGET$, (ii) $p_j$'s status is $GRANT$ and its state encodes that the grant is already given to $p_i$, or (iii) $p_j$'s status is $GET$ ($p_i$ and $p_j$ have concurrent requests) and $p_i$ has the higher priory according to $\precedes(m)$. Otherwise, $p_j$ replies to $p_i$ with a $\messageForm{DENY}$ message. 


\end{requirement}

\begin{requirement}[Releasing] 
	\label{req:explicitRelease}
	Let $p_i$, $p_j$, $\D$, and $M$ be defined as in Requirement~\ref{req:GrantingDenying}. We require agent $p_i$ to multicast a release message, $\messageForm{RELEASE}$, to all requestee agents $p_j$ in $\D$ only in the following cases: (i) $p_i$'s status is $GET$ and it receives a message $\messageForm{DENY}$ from $p_j \in \D \cap M$, (ii) $p_i$'s status is $GET$ and its timer expires so that it needs to defer its requests by changing status to $TRYGET$, (iii) $p_i$'s status is $EXECUTE$ and it has just left the intersection, and (iv) case (iii) of Requirement~\ref{req:GrantingDenying} occurs. We also require agent $p_j$ to call $\ntfyEstimat(\revoke)$ only in two special occasions. The first occasion is when $p_j$ status is $GRANT$, its state encodes that it has granted permission to agent $p_i$ and it has just received the message $\messageForm{RELEASE}$ from agent $p_i$. The second occasion occurs when $p_j$'s status is $EXECUTE$ and $p_j$ has left the critical section. 
\end{requirement}
	  
\begin{figure*}[t!]
	\centering
	\includegraphics[page=1,scale=0.5]{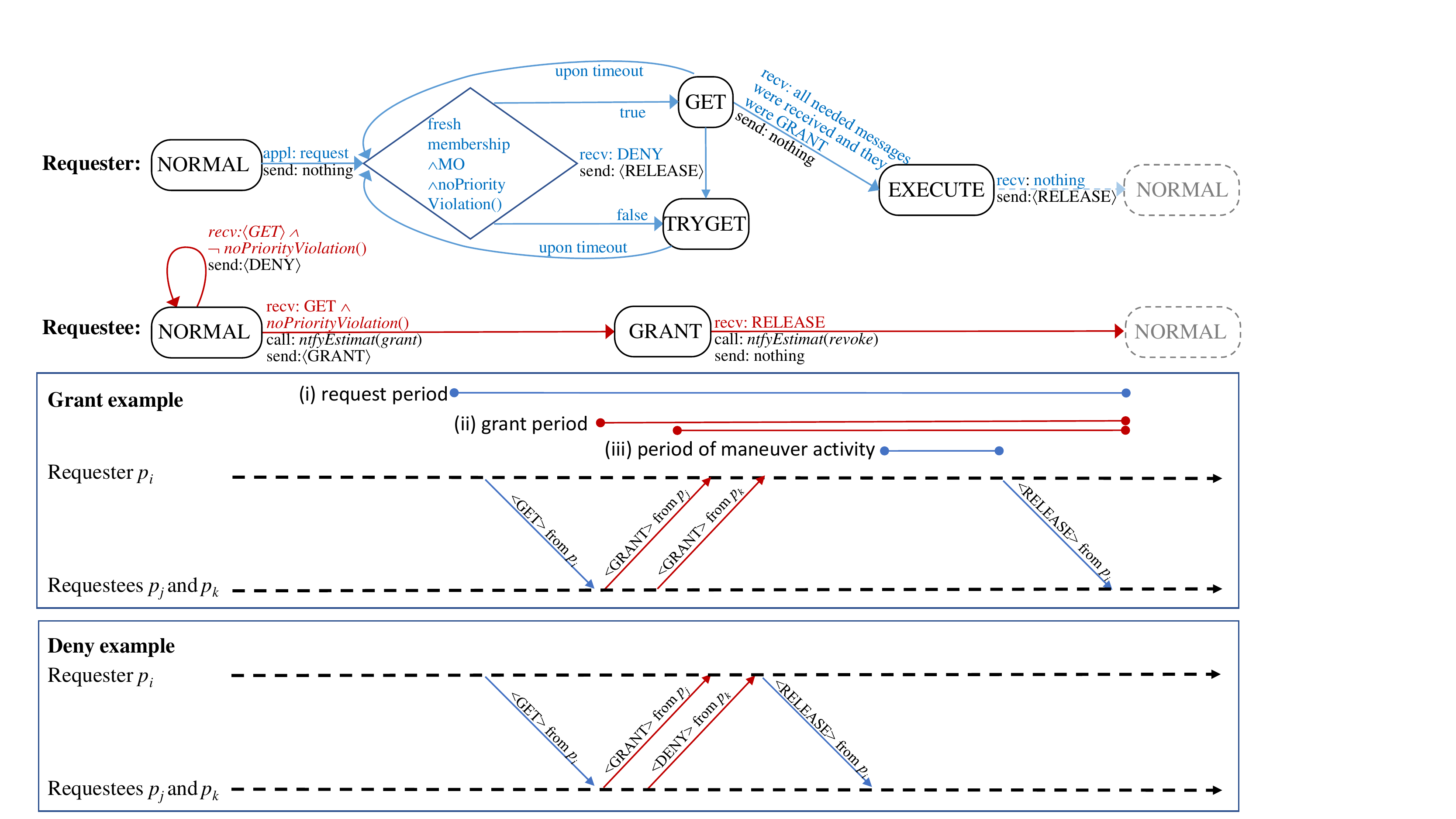}
	\caption{\label{fig:agentDialog}The dialog between the requester and the requestee depicted by the state-machines above and below respectively. Examples for grant and deny processes are described in the message sequence charts provided above and below respectively. The grant example also shows the periods in which (i) the request is ongoing, (ii) the grant has changed the state of the requestees so that the requester has the priory over the requestees, and (iii) when the manoeuvre is performed.}
\end{figure*}

Consider an agent $p_i$ that has requested a permission from all agents in $\D$ (when the value of $p_i$'s membership set was $\D$, cf. Requirement~\ref{req:RequestInitialisation}). Suppose that $G$ is the set of agents whose responses arrived to $p_i$ from $p_j \in \D$ within $2T_D-T_A$ time units, and $M$ is $p_i$'s current membership. We say that $p_i$ is \emph{fully-granted} when (i) $G \supseteq \D \cap M$ and (ii) all the agents in $G$ responded with a $\messageForm{GRANT}$ or if at least $2T_D$ time units have passed since the request was first sent (to the agents in $\D$). Whenever $G$ includes an agent that responded with a $\messageForm{DENY}$, we say that $p_i$ is \emph{fully-denied}. In all other cases, we say that $p_i$'s request is \emph{pending}.

\begin{requirement}[Allowing and Deferring Manoeuvres] 
	\label{req:AllowingDeferring}
We require the agent $p_i$ to: (i) only perform a manoeuvre that is fully-granted and to send a release message upon its completion, and (ii) in the case of a fully-denied manoeuvre, send a release message to all agents in $\D$ and restart the request process after $T_A$ time units.
\end{requirement}

Suppose that there is a (possibly unknown, yet finite) bound, $T_{admissible}$, such that a period of $T_{admissible}$ time units in which $p_i$'s membership is fresh and empty implies (i) a safe crossing opportunity for $p_i$ during that period (in the sense that $p_i$'s membership is empty, the manoeuvre opportunity indicator is $\true$ and it has not granted any request for a period that overlaps with this period of $T_{admissible}$ time units) and that (ii) agent $p_i$ has the time to perform the manoeuvre (in the sense that $p_i$ has the time to notice that its membership is fresh and empty, and the manoeuvre opportunity indicator is $\true$). We call such a crossing opportunity an \emph{admissible crossing opportunity}. We note that the fulfilment of Requirement~\ref{req:MembershipService} implies the existence of $T_{admissible}$.


	%
	%

\begin{requirement}[Progression] 
	\label{req:Progression}
	Let $E$ be a synchronous system execution in which admissible crossing opportunities occur infinitely often (for every direction of the intersection). Moreover, suppose that, infinitely often, the agents that are the closest to the intersection, aim at crossing the intersection. Then, a fully-granted agent crosses the intersection infinitely often. 
\end{requirement}

%

The constants considered in the proposed algorithms and that are mentioned in the Theorems and Lemmas are provided in Table~\ref{tab:constants}.

\section{Algorithms}

\label{sec:Pseudo}

The time constants used in the proposed algorithms are provided in Table~\ref{tab:constants}. The following internal size order is implied: $T_{Man} > T_D \geq T_A$ and $T_M>T_A$.

\begin{table}[!h]
	\centering
	\begin{tabular}{|l|l|} 
		\hline
		$T_M$ & Period of membership protocol 
		\\ \hline
		$T_A$ & Period of agent registry update 
		\\ \hline
		$T_D$ & Upper bound on transmission delay between vehicles 
		\\ \hline
		$T_{Man}$ & Upper bound on manoeuvre execution time 
		\\ \hline
	\end{tabular}
	\caption{Time constants used in Algorithm \ref{alg:membership} and \ref{alg:agent-events}. Observe that we assume that $T_{Man} > T_D \geq T_A$ and $T_M>T_A$.} 
	\label{tab:constants}
\end{table}


\begin{algorithm}[ht!]
	\begin{\VCalgSize}
		
		\textbf{Structures:}\\
		$Segment = (xs, xe)$\tcp*[r]{positions: start, end}
		$AgentState = (ta,x,v,a)$\tcp*[r]{timestamp, position, velocity, acceleration}
		$Agent = (\aID,AS)$\tcp*[r]{agent ID, AgentState structure}
		$Membership = (tm,MO,\SM)$\tcp*[r]{timestamp, MO flag, $\SM$ IDs of agents in the membership}
		\smallskip
		\textbf{Constants:}\\
		$\sID$: the segment ID\;
		$S := \getSegmentData(\sID)$\tcp*[r]{Segment registry}
		\smallskip
		\textbf{Variables:}\\
		$\mathcal{MR} \gets \emptyset$\tcp*[r]{Membership registry}
		\ \\
		\textbf{Interfaces:}\\
		$\clock()$: read current time\;
		$\getSegmentData(sID)$: read Segment registry for segment $sID$ from the storage service\;
		$\getARSet(sID)$: read all Agent Registries for segment $sID$ from the storage service\;
		$\storeMR(aID,MR)$: write $MR$ (Membership registry) for agent $aID$ on the storage service\;
		$\getUnsafeAgents (AR,t)$: uses $\noPriorityViolationStat$ (Section \ref{sec:unsafeAgents}) to obtain the sets of agents with whom the agent $\aID$ identified by $AR$ may cause a priority violation for the different possible turns defined by $\getPossibleTurns(AR)$. We assume that this set includes all the agents of higher priority than $\aID$ that could possibly arrive to the intersection by the time $t$\;
		$\getReachableAgents(AR,t)$: obtain the set of agents that will be within communication range from the agent identified by $AR$ until $t$.  \;
		$\getPossibleTurns(AR)$: obtain a list of possible turns for an agent\;
		\ \\
		
		\textbf{do forever (once in every $T_M$ time units)\label{ln:T_UPDATEExp}}
		\Begin{
			\textbf{let} $\mathcal{A} = \getARSet(\sID())$ \label{ln:MembershipGetARset}\;
			\ForEach{$AR$ in $\mathcal{A}$}{
				\If{$AR.AS.x > S.xs \land AR.AS.x < S.xe$ \label{ln:MembershipDistRestriction}}{ 
					\textbf{let} $t = \clock()$\;
					\textbf{let} $\mathcal{U}=\getUnsafeAgents (AR,t+2T_M+2T_D+T_{MAN})$ \label{ln:MembershipUnsafe}\;
					\textbf{let} $\mathcal{R}=\getReachableAgents(AR,t+2T_M+2T_D+T_{MAN})$ \label{ln:MembershipReachable}\;
					\ForEach{$turn$ \textnormal{in} $\getPossibleTurns(AR)$}{
						\textbf{let} $t_{min} = min(\alpha.AS.ta : ( \alpha \in \mathcal{A}) \land (\alpha.aID \in \mathcal{U}[turn]))$ \label{ln:MembershipClock}\;
						\leIf{$\mathcal{U}[turn] \subset \mathcal{R}$ \label{ln:MembershipUnsafeReachable}}{$MR[turn] \gets (t_{min},\true, \mathcal{U})$}{$MR[turn] \gets (t_{min},\false,\emptyset)$}
					}	
					$\storeMR(AR.aID,
					MR)$ \label{ln:StoreMembership}\; 
				}
			}
		}
	\end{\VCalgSize}
	\caption{The membership algorithm.}
	\label{alg:membership}
\end{algorithm}

\begin{figure*}[ht!]
	\LinesNotNumbered
	\fbox{
	\begin{minipage}{\textwidth}

		    \textbf{Data types:}\\
            \hspace*{1em} $MessageType \in \{GET, DENY, GRANT, RELEASE\}$ \\
            \hspace*{1em} $StatusType \in \{NORMAL, GET, GRANT, TRYGET, GRANTGET, EXECUTE \}$ \\		
            \smallskip
			\textbf{Structures:} \\
			$AgentState = (ta,x,v,a)$\tcp*[r]{timestamp, position, velocity, acceleration}
			$Agent = (\aID,AS)$\tcp*[r]{agent ID, AgentState structure}
			$AgentTag = (ts,\aID,turn)$\tcp*[r]{timestamp, agent ID, intended turn}
			$Membership = (tm,MO,\SM)$\tcp*[r]{timestamp, MO flag, $\SM$ set of agent IDs}
			$Message = (senderID,t,type,(data))$\tcp*[r]{agent ID,timestamp,message type,message data}
			\smallskip
			\textbf{Constants:} \\
			$\aID$: a hard-coded agent identifier\; \\
			\smallskip
			\textbf{Variables:} \\
			$AR \gets \emptyset$ \tcp*[r]{Agent registry} 
			$MR \gets \emptyset$\tcp*[r]{Membership registry, set of Memberships}
			$status \gets NORMAL$ \\
			$tag \gets \bot$\tcp*[r]{Unique tag for each manoeuvre}
			$grantID \gets \bot$\tcp*[r]{ID of the agent with an active GRANT}
			$\M \gets \emptyset$\tcp*[r]{Set of received messages}
			$\D \gets \emptyset$\tcp*[r]{Set of destination agent IDs}
			$\mathcal{R} \gets \emptyset$\tcp*[r]{Set of expected response agent IDs}
			$\tRETRY$\tcp*[r]{Timer for keeping track of when to next try starting a new request round}
			
			\ \\
			\textbf{Interfaces:}\\
			$\clock()$: read current time\; \\
			$\position()$: read current agent position\; \\
			$\velocity()$: read current agent velocity\; \\
			$\acceleration()$: read current agent acceleration\; \\
			$\startTimer(timer,delay)$: start named $timer$ for $delay$ time units\; \\
			$\stopTimer(timer)$: stop named $timer$\; \\
			$\storeAR(AR)$: write an $AR$ registry to the storage service \; \\
			$\getMR(aID)$: read from the storage service the $MR$ for agent $aID$\; \\
			$\doManeuver$: execute the manoeuvre by setting the intention to go\; \\
			$\noPriorityViolationDyn(AR,turn)$: check if the agent specified in $AR$ have sufficient time to cross the intersection without violating the priority of the ego vehicle, see Section \ref{sec:nPV_Dynamic}\; \\
			$\last()$: returns $\true$ if $\mathcal{R} \cap D \cap MR.\SM = \emptyset$, $\false$ otherwise\; \\
			$\precedes(m)$: returns $\true$ if $p_{\aID}$ has an on going request and that request has the priority over the request of message $m$, eg having an earlier time stamp, or if equal, lower agent ID.\\
			%
			%
			$\hasLeftCriticalSection(aID)$: returns $\true$ when the agent $p_{aID}$'s location is known to be outside the boundaries and heading out from the intersection.\;  \\
			$\ntfyEstimat(\grant,aID/\revoke)$: Notify the risk estimator about priority changes of either a grant of agent $aID$ or the revoke of a grant\;
			
			
		
		\caption{\label{alg:defs}Structures, variables and interfaces for the manoeuvre negotiation protocol}
		\end{minipage}
}

\end{figure*}

\begin{algorithm}[ht!]
	\begin{\VCalgSize}
		
		\textbf{do forever (once in every $\TdoForeverLoop$ time units)\label{ln:doForever}} 
		\Begin{
			
			

			$AR \gets (\aID,(\clock(), \position(), \velocity(), \acceleration()))$ \label{ln:getOwnStatus}\;
			$\storeAR(AR)$\label{ln:storeAR}\;
			$MR \gets \getMR(\aID)$ \label{ln:getMRdoForever}\;
			$\M \gets \{m \in \M: m.senderID \in (\D \cap MR.SM) \}$\label{ln:removeStale}\;
			\If(\tcp*[f]{All expected answers received}){$status=GET \land \last()$\label{ln:statusGETlast}}{
				$\stopTimer(\tRETRY) \label{ln:doForeverStopTRETRY}$\tcp*[r]{No need for a retry}
				\lIf{$\nexists {m \in \M}:m.type = DENY$\label{ln:nexistsMtypeDENY}}{
					$status \gets EXECUTE$; $\ntfyEstimat(\grant)$
				}\lElse{
					\{$status \gets TRYGET$;
					\textbf{multicast} $\langle \aID, \clock(), RELEASE, (AR, tag) \rangle$ \textbf{to} $\D$\label{ln:multReleaselast};					$\startTimer(\tRETRY,T_A)$\label{ln:nexistsMtypeDENYelse};\}}
			}
			\If{$status=EXECUTE \land \hasLeftCriticalSection(\aID,MR.\SM)$ \label{ln:EXECUTEHasLeftCS}}{
				\textbf{multicast} $\langle \aID, \clock(), RELEASE, (AR, tag) \rangle$ \textbf{to} $\D$\tcp*[r]{Explicit release}
				$status \gets NORMAL$; $\ntfyEstimat(\revoke)$;\label{ln:statusgetsNORMAL}
			}
			\If{$status \in {GRANT, GRANTGET} \land \hasLeftCriticalSection($grantID$,MR.\SM)$ \label{ln:GRANTHasLeftCS}}{
			    $\ntfyEstimat(\revoke)$\label{ln:revoke}; $grantID \gets \bot$\label{ln:GRANTClearID}\;
				\lIf{$status = GRANTGET$}{$status \gets TRYGET$; $\tryManeuver(tag.turn)$\label{ln:prevStatusGRANTGETtryManeuver}}
				\lElseIf{$status = GRANT$}{$status \gets NORMAL$\label{ln:GRANTtoNORMAL}}
			}
		}

		\textbf{procedure} $\tryManeuver(turn)$ \label{ln:tryManeuverTurn}
		\Begin{				                              
			\lIf{$status \in \{NORMAL, GRANT\}$}{$tag \gets (\clock(),\aID,turn)$ \label{ln:createTag}}
			\uIf{$status \in \{NORMAL, TRYGET\}$\label{ln:statusNORMALTRYGET}}{
				\textbf{let} $ts = \clock()$\;
				$(status,AR,MR) \gets (TRYGET,(\aID, (ts,\position(),\velocity(),\acceleration()),\getMR(\aID)[turn])$\label{ln:statusARMR}\;
				\uIf{$(ts < MR.tm + 2T_M \land MR.MO)$\label{ln:tsMRturntm2TMMRMO}}{
					$(\M,\D,\mathcal{R}, status) \gets (\bot,MR.\SM,MR.\SM, GET)$\label{ln:statusGET}\;
					\textbf{multicast} $\langle \aID, ts, GET, (AR, tag) \rangle$ \textbf{to} $\D$\;
					$\startTimer(\tRETRY,2T_D)$\label{ln:TRETRY2TD}\;
				}
				\lElse{$\startTimer(\tRETRY,T_A)$\label{ln:startTimerTRETRYTA}}
			}
			\lElseIf{$status=GRANT$}{$status \gets GRANTGET$\label{ln:statusGRANTstatusGRANTGET}}
		}
		\textbf{upon} timer $\tRETRY$ expires \label{ln:timerRETRYexpires}
		\Begin{
			\If{$status=GET$}{
				$status \gets TRYGET$;
				\textbf{multicast} $\langle \aID, \clock(), RELEASE, (AR, tag) \rangle$ \textbf{to} $\D$\label{ln:retryTimeRelease};
				$\startTimer(\tRETRY,2T_D)$; \label{ln:TRETRYRelease}}
			$\tryManeuver(tag.turn)$\label{ln:tryManeuvertagturn}\;
		}
		
			
		
		\textbf{upon} message $m$ received at time $tr$ \label{ln:MessageProcessing}
		\Begin{
			\If(\tcp*[f]{Message is timely}){$tr-m.t \leq T_D$ \label{ln:MeassageAge}}{
				\lIf{$m.type \in \{GRANT,DENY\} \land status=GET$ \label{ln:RecieveDeny}}{
					$(\M,\mathcal{R}) \gets (\M \cup \{m\},\mathcal{R} \setminus \{m.AR.aID\} )$
				}
				\uElseIf{$m.type=GET$ \label{ln:MessageGET}}{
					\uIf{$\noPriorityViolationDyn(requestee, m.AR) \land (status\in\{NORMAL,TRYGET\}  \lor (status\in\{GRANT,GRANTGET\} \land grantID = m.AR.aID) \lor (status=GET \land \neg \precedes(m)))$ \label{ln:GRANTDENYCondition}}{
						\lIf{$status=NORMAL$ \label{ln:MessageGETNORMAL}}{$(status,grantID) \gets (GRANT,m.AR.\aID)$}
						\ElseIf{$status \in \{GET,TRYGET\}$ \label{ln:MessageGETGETTRYGET}}{
							$\stopTimer(\tRETRY)$\label{ln:GETTRYGETstoptime}\;
							\lIf{$status=GET$}{\textbf{multicast} $\langle \aID, \clock(), RELEASE, (AR, tag) \rangle$ \textbf{to} $\D$ \label{ln:ReleaseWhenGET}}
							$(status,grantID) \gets (GRANTGET,m.AR.\aID)$\label{ln:GETTRYGETsave}\;
						}
						$\ntfyEstimat(\grant, m.AR.\aID)$; \textbf{send} $\langle \aID, tr, GRANT, \bot \rangle$ \textbf{to} $m.AR.\aID$ \label{ln:ntfyEstimatGrant}\;
					}
					\lElse{\textbf{send} $\langle \aID, tr, DENY, \bot \rangle$ \textbf{to} $m.AR.\aID$ \label{ln:SendDENY}}
				}
				\ElseIf{$m.type=RELEASE \land status \in \{GRANT,GRANTGET\} \land grantID = m.AR.\aID$}{
					$\ntfyEstimat(\revoke)$; $grantID \gets \bot$\ \label{ln:ReleaseNtfyEstimatRevoke}\;
					\lIf{$status=GRANT$}{$status \gets NORMAL$\label{ln:GRANTgetsNORMAL}}
					\lIf{$status=GRANTGET$}{
						$status \gets TRYGET$;
					    $\tryManeuver(tag.turn)$\label{ln:endMessageProccessing}
					}
				}
			}
		}
		
	\end{\VCalgSize}
	\caption{The Manoeuvre Negotiation protocol}
	\label{alg:agent-events}
\end{algorithm}


\Section{Algorithm description}
\label{sec:AlgorithmDescription}

This section provides a line by line description of the Membership Algorithm (Algorithm \ref{alg:membership}) and Manoeuvre Negotiation Algorithm (Algorithm \ref{alg:agent-events}).

\Subsection{Description of the Membership Algorithm}
\label{sec:AlgDescMembership}

The \textbf{Membership Algorithm} (Algorithm \ref{alg:membership}) periodically updates the membership set for every vehicle in a specified geographical section called a \emph{segment}. A \textbf{membership set} holds multiple memberships. A membership in the membership set for a vehicle $p_i$ consists of a set called Safety Membership ($\SM$) made up of the IDs of all other vehicles of higher priority in the same segment of the road, a flag called Manoeuvre Opportunity ($MO$) that indicates whether or not the membership is valid, and a time stamp. There has to be one membership per manoeuvre $p_i$ can perform in the segment since traffic priorities are manoeuvre-based and we want the protocol to preserve the initial traffic priorities if possible. The segment in this case is a 4-way intersection. The possible manoeuvres therefore correspond to turning left, going straight, and turning right.

The Membership Algorithm provides a periodic update of memberships every $T_M$ time units (line \ref{ln:T_UPDATEExp}). The procedure starts by retrieving all agent registries of the vehicles in the selected segment from a storage server (line \ref{ln:MembershipGetARset}). The algorithm iterates through these agent registers and for each agent $p_i$ checks if $p_i$ is still in the section (line \ref{ln:MembershipDistRestriction}). If so, the algorithm calculates which agents $\U$ have higher priority than $p_i$ and may reach the intersection in the time it would take for $p_i$ to reach and perform a manoeuvre through the same intersection (line \ref{ln:MembershipUnsafe})\footnote{For more details see Section \ref{sec:unsafeAgents}}. $\U$ will have 3 entries for the 3 possible manoeuvres described above. In addition to $\U$, the set of all vehicles $\R$ that are within communication distance from $p_i$ are determined (line \ref{ln:MembershipReachable}). The earliest time stamps $t_{min}$ of the agent states in $\U[turn]$ for each specific turn are also determined (line \ref{ln:MembershipClock}). These time stamps will be used as a measure of how current the memberships are. A $t_{min}$ far back in time indicates that the membership was calculated using old data and that the membership may be too old to be used safely. 

The next step is to determine the value of the $MO$s for the manoeuvre specific memberships. $MO$ is set to $\true$ \textit{if and only if} all agents in $\U$ are also in $\R$ (line \ref{ln:MembershipUnsafeReachable}), which implies that all agents that $p_i$ can cause a priority violation for are also reachable by $p_i$. It's also possible to use the $MO$-flag to indicate that the provided memberships are invalid due to, for example, limited capacity of the membership service. 

The value assigned to $\SM$ depends on the value of $MO$. $\SM$ is assigned the value of $\U[turn]$ for the specific manoeuvre described by $turn$ if $MO=true$, and is assigned the empty set if $MO=false$ (line \ref{ln:MembershipUnsafeReachable}). The use of $MO$ thereby provides a way to  both reduce the number of sent messages -- since vehicles will not send a request if $MO=false$ -- and a way to indicate if an empty membership is valid.

Finally, all memberships in the membership set are stored on the storage server accessible by all agents (line \ref{ln:StoreMembership}).

\Subsection{Description of the Manoeuvre Negotiation Algorithm}
\label{sec:algDescrMNA}

The \textbf{Manoeuvre Negotiation Algorithm} (Algorithm \ref{alg:agent-events}) consists of four main components: the core state-update procedure that is called periodically (lines \ref{ln:doForever} to \ref{ln:statusgetsNORMAL}), a procedure that is called when an agent seeks permission to perform a manoeuvre (lines \ref{ln:tryManeuverTurn} to \ref{ln:statusGRANTstatusGRANTGET}), a timer-triggered procedure used to ensure algorithm progression in the event of network failures (lines \ref{ln:timerRETRYexpires} to \ref{ln:tryManeuvertagturn}), and a message processing procedure(lines \ref{ln:MessageProcessing} to \ref{ln:endMessageProccessing}).

The \textbf{state-update procedure} itself has four main parts, the first of which handles communication with the storage server. Every vehicle in our system measures their current position, heading,  velocity, and acceleration every $T_A$ time units (line \ref{ln:getOwnStatus}) and then stores that information on the storage server (line \ref{ln:storeAR}). This vehicle information is used by the membership service to calculate memberships in Algorithm \ref{alg:membership}. Every vehicle also fetches the latest version of their membership set from the storage server and stores it locally (line \ref{ln:getMRdoForever}).

The second part of the procedure determines if a vehicle that is requesting a permission to perform a manoeuvre -- a \emph{requester} --  has been fully granted. Replies from vehicles that are not a part of the intersection between the membership that was valid when the request was initiated ($\D$) and the current membership ($MR.\SM$) are filtered out from the set of received replies $\M$ (line \ref{ln:removeStale}). Vehicles that in $\D$ but not $MR.\SM$ must have already left the intersection, so their replies are considered irrelevant. Vehicles that are in $MR.\SM$ but not $\D$ are considered too far away to have any impact on the requested manoeuvre. Vehicles in neither set are also either too far away, or have at most equal priority to the requester and therefore should not be allowed to affect the decision process. 

Next, the requesters -- having the status $GET$ -- confirm that all expected replies to the sent request have arrived by calling $last()$ (line \ref{ln:statusGETlast}). If all expected replies have arrived then the $\tRETRY$ is stopped to mark the end of the current request round (line \ref{ln:doForeverStopTRETRY}). In order to assess if the requester is fully granted, $\M$ is searched for any $DENY$ replies and, if none are found, the requester changes status to $EXECUTE$ and is thereby free to enter the intersection (line \ref{ln:nexistsMtypeDENY}). Otherwise the requester will send a $RELEASE$ message to all vehicles in $\D$ telling them to stop considering the requester granted and the $\tRETRY$ timer will be started again to continue the request into a new round (line \ref{ln:multReleaselast}).

The third part of the procedure determines when a fully granted vehicle should release its grant. A fully granted agent $p_i$ -- having status $EXECUTE$ -- will use their measured position and knowledge regarding the intersection's geometry to determine if it has left the intersection (line \ref{ln:EXECUTEHasLeftCS}) and should therefore send a $\messageForm{RELEASE}$ of the grant to every vehicle in $\D$ (line \ref{ln:EXECUTEHasLeftCS}). The agent $p_i$ will then return to status $NORMAL$ and call $\ntfyEstimat(\revoke)$ to also inform the RE about the revoked grant (line \ref{ln:statusgetsNORMAL}). 

The last part of the state-update procedure handles implicit release of a grant which enables a grant to be released even when the explicit release message is lost. An agent that is currently granting another agent, having status $GRANT$ or $GRANTGET$, checks if the latest measured position of the granted agent lies outside and is heading away from the intersection (line \ref{ln:GRANTHasLeftCS}). If the condition is $\true$ then the agent notifies the RE about the priority change by calling $\ntfyEstimat(\revoke)$ and clears $grantID$ (\ref{ln:GRANTClearID}). An agent previously prevented from starting a request round, denoted by having status $GRANTGET$, takes on status $TRYGET$ and calls $\tryManeuver$ to try to start a new request round (line \ref{ln:prevStatusGRANTGETtryManeuver}). An agent with status $GRANT$ instead goes back to status $NORMAL$ (line \ref{ln:GRANTtoNORMAL}). 

The request initiating procedure $\tryManeuver$ is invoked when starting a new round of a requests. During the first round of a request, where the vehicle has status $NORMAL$ or $GRANT$, the procedure generates an identification tag for the request, then stores: the time stamp of when the first request round started, the ID of the requester, and the manoeuvre the requester wants to perform (line \ref{ln:createTag}). Requesters that are not currently granting another agent, having status $NORMAL$ or $TRYGET$, proceed to measure their current agent state and fetch the latest membership from the storage server (line \ref{ln:statusARMR}). 

Starting a request round further requires a valid membership, implying that the membership is based on data no older than $2T_M$ time units and with a $MO=true$ (line \ref{ln:tsMRturntm2TMMRMO}). If the membership is valid, then the state variables $\M$, $\D$, and $\R$ will be assigned the values $\bot$, $MR.\SM$, and $MR.\SM$ respectively, where $\R$ is the set of agents that are expected to reply and the status is changed to $GET$ (line \ref{ln:statusGET}). A $\messageForm{GET}$ request is then sent to all agents in $\D$ along with the request tag. The $\tRETRY$ timer is set to $2T_D$ to limit the duration of the request round (line \ref{ln:TRETRY2TD}). An agent without a valid membership will wait for a new membership from the membership server by starting the $\tRETRY$ timer and try starting a request round at a later point in time \ref{ln:startTimerTRETRYTA}. A potential requester that is currently granting another agent, having status $GRANT$, is not allowed to start a request round until the grant is released, but a status change to $GRANTGET$ is made to signify that the agent wants to send out a request later when it is allowed (line \ref{ln:statusGRANTstatusGRANTGET}).

Upon expiry of the $\tRETRY$ timer, an agent currently in a request round, having status $GET$, will end it by sending a $RELEASE$ to all agents in $\D$, changing its status to $TRYGET$ and starting the $\tRETRY$ timer again (line \ref{ln:retryTimeRelease} and \ref{ln:TRETRYRelease}). The agent did not get all expected replies in the set time frame and must therefore stop the current round to avoid deadlocks. If the timer expires for an agent without an ongoing request round, having status $TRYGET$, then the agent will try to start a new request round by calling the $\tryManeuver$ procedure (line \ref{ln:tryManeuvertagturn}).

The final part of the Manoeuvre Negotiating algorithm determines how incoming messages are processed depending on the agent's current status and the agent state of the sender. Messages sent more than $T_D$ time units ago are discarded in order to avoid messages from previous request rounds interfering with the current round (line \ref{ln:MeassageAge}). 

The first type of messages that are processed are request replies on the form \\$\messageForm{GRANT}$ and $\messageForm{DENY}$. Any requester will store these messages in the set $\M$ and remove the sender ID from the set of agents whose expected responses have not yet arrived, $\R$, (line \ref{ln:RecieveDeny}). 

The second message type that is processed is $\messageForm{GET}$ which symbolises a permission request (line \ref{ln:MessageGET}). The receivers have to decide if the request should be granted or denied. Firstly, a call to $\noPriorityViolation$ determines if the sender can, while following normal driving patterns, cross the intersection without blocking the receiver (line \ref{ln:GRANTDENYCondition}). In addition to the non-blocking criteria, the receiver has to check its own status to see if it allows the agent to grant the sender. 

Status $NORMAL$ and $TRYGET$ imply that the agent does not have an active request round and is therefore allowed to grant another agent (line \ref{ln:GRANTDENYCondition}). 

Agents with status $GRANT$ or $GRANTGET$ are currently holding a grant for another agent. If this is the same agent as the sender then the receiver is allowed to reply with a $\messageForm{GRANT}$ (line \ref{ln:GRANTDENYCondition}). 

The last status that allows the receiver to reply with a grant occurs where the receiver both has status $GET$ --  meaning that it has an active request round -- and the sender has higher priority than the receiver -- determined by a call to the interface $\precedes$ (line \ref{ln:GRANTDENYCondition}). The interface $\precedes$ will return $\false$ if either the sender's request was started before the receiver's request or if they were started concurrently and the sender's ID is lower than the receiver's. 

A receiver fulfilling both the non-blocking and status criteria will store the ID of the granted agent if it has not yet done so, and change its status to either $GRANT$ -- if it has status $NORMAL$ -- or to $GRANTGET$ -- if it has status $GET$ or $TRYGET$ (lines \ref{ln:MessageGETNORMAL}, \ref{ln:MessageGETGETTRYGET}, and \ref{ln:GETTRYGETsave}). Changing from status $GET$ or $TRYGET$ to $GRANTGET$ also implies both that the $\tRETRY$ timer has to be stopped and that any ongoing round must be stopped by sending a $\messageForm{RELEASE}$ to all agents in $\D$, since an agent in $GRANTGET$ is not allowed to either start or have an ongoing request round (lines \ref{ln:GETTRYGETstoptime} and \ref{ln:ReleaseWhenGET}). The risk estimator is then notified with $\ntfyEstimat(\grant)$ that the priorities have changed and a $\messageForm{GRANT}$ reply is sent to the request sender (line \ref{ln:ntfyEstimatGrant}).
If these criteria are not fulfilled, a $\messageForm{DENY}$ is sent in reply (line \ref{ln:SendDENY}).  

The last type of messages to be processed are $\messageForm{RELEASE}$ messages. An agent with status $GRANT$ or $GRANTGET$ that receives a release message from the holder of the grant will call $\ntfyEstimat(\revoke)$ to inform the risk estimator of the expiry of the granted vehicle's higher priority (\ref{ln:ReleaseNtfyEstimatRevoke}). A $GRANT$ agent will return to status $NORMAL$ while a $GRANTGET$ agent switches to status $TRYGET$ and calls $\tryManeuver$ in order to try to start a request round (line \ref{ln:GRANTgetsNORMAL} to \ref{ln:endMessageProccessing}).

\Section{Correctness Proof}
\label{sec:corrProof}
By the stature of Algorithm~\ref{alg:agent-events}, we can observe the correctness of Corollary~\ref{thm:timerConsistency}.

\begin{corollary}
\label{thm:timerConsistency}
Let $E$ be an execution of the proposed solution. At any system state, it holds that 
$(status \in \{TRYGET\} \implies  active(\tRETRY))$, where $active(t)$ is a predicate that represents the fact that timer $t$ has been started (by the function $\startTimer(t)$), and has not stopped (by the function $\stopTimer(t)$) or expired without triggering the event bounded to timer $t$. Furthermore, an active $\tRETRY$ timer expires within $B\_RETRY=\max\{2T_D, T_A\}$ time units.
\end{corollary}

\begin{lemma}[Membership] 
\label{thm:Membership}
Let $E$ be a synchronous execution of algorithms~\ref{alg:membership} and~\ref{alg:agent-events}. Requirement~\ref{req:MembershipService} holds throughout $E$. 
\end{lemma}
\begin{proof}
To prove that Requirement \ref{req:MembershipService} holds throughout $E$ it needs to be shown that a fresh membership $M$ of an agent $p_i$ with $MO=\true$ contains \emph{all} agents that could potentially cause a priority violation with $p_i$ from the current time $T$ until $T + 2T_D + T_{MAN}$. We say that a membership is \emph{valid} if this property holds for the membership. We refer to all other memberships as non-valid. It needs to be proven that, when $t$ is the timestamp bound to $M$ (line~\ref{ln:MembershipClock}), $M$ can only be non-valid if $MO=\false$ or $M$ is no longer fresh, i.e., $t \geq T - 2T_M$.


Algorithm~\ref{alg:membership} (lines~\ref{ln:T_UPDATEExp} to~\ref{ln:StoreMembership}) defines the procedure in which memberships are created and stored. We review this procedure and show its key properties.

\emph{Algorithm~\ref{alg:membership} considers all potential priority violations.~~}
The procedure calls \\ $\noPriorityViolationStat(t + 2T_M + 2T_D + T_{MAN})$ to determine all agents $\U$ that have higher priority than $p_i$ and that could potentially reach the intersection at some point from $t$ to $t + 2T_M + 2T_D + T_{Man}$ (line \ref{ln:MembershipUnsafe}). These agents $\U$ are the ones that could potentially lead to a priority violation with $p_i$ according to our definition in Section \ref{sec:prioViolation}. 

\emph{Algorithm~\ref{alg:membership} assigns correct $MO$ values.~~}
For a given turn, the membership $M$ is assigned the set $\U[turn]$ if, and only if, $MO$ is assigned the value $\true$, otherwise $M$ is assigned the empty set (line  \ref{ln:MembershipUnsafeReachable}). The variable $MO$ thereby distinguishes a valid empty membership set from a non-valid empty membership set.

\emph{Algorithm~\ref{alg:membership} assigns correct time stamps.~~}
The time stamp associated with a membership is assigned the value of the earliest time stamp $t_{min}$ of the agents in $\U[turn]$ (line \ref{ln:MembershipClock}). This implies that a fresh membership is valid from $T$ to $T + 2T_D + 2T_M + T_{MAN} - (T-t_{min})$ where $T-t_{min} \leq 2T_M$ due to the freshness of $M$. This proves Requirement \ref{req:MembershipService} since we assume that vehicles do not drive backwards on highways.
\end{proof}

\begin{theorem}[Safety] 
	\label{thm:Safety}
	Let $E$ be an execution of algorithms~\ref{alg:membership} and~\ref{alg:agent-events} that is not necessarily synchronous in which Requirement~\ref{req:GrantingDenying} holds throughout. Requirements~\ref{req:RequestInitialisation},~\ref{req:GrantingDenying}, and~\ref{req:AllowingDeferring} hold throughout $E$. 
\end{theorem}
\begin{proof}
\begin{enumerate}
	\item \textbf{Requirement~\ref{req:RequestInitialisation}:~~} By the structure of Algorithm~\ref{alg:agent-events},
	%
	%
	we can observe that an agent $p_i$ that aims to perform a manoeuvre, has to call $\tryManeuver()$ (line~\ref{ln:tryManeuverTurn}) and be granted a permission before reaching the status  $EXECUTE$. In detail, $\tryManeuver()$ is the only interface procedure for agent $p_i$, which includes lines~\ref{ln:statusGET} to~\ref{ln:TRETRY2TD} in which requests are sent. By line~\ref{ln:tsMRturntm2TMMRMO}, we see that $p_i$ sends such a request only when (i) its membership is fresh and (ii) there is an positive indication for a manoeuvre opportunity. Otherwise, $p_i$ checks conditions (i) to (iii) after $T_A$ time units (line~\ref{ln:startTimerTRETRYTA}). Thus, Requirement~\ref{req:RequestInitialisation} is fulfilled.

	\item \textbf{Requirement~\ref{req:GrantingDenying}:~~}

    Algorithm~\ref{alg:agent-events} defines how a $\messageForm{GET}$ message is sent (lines~\ref{ln:tsMRturntm2TMMRMO} to~\ref{ln:TRETRY2TD}) and received (line~\ref{ln:MessageGET} to~\ref{ln:SendDENY}). We note that $p_i$ sends a $\messageForm{GET}$ message only if the if-statement condition at line~\ref{ln:GRANTDENYCondition}, i.e., $(t_i<M.tm+2T_M \land MR.MO)$ which implies that the first assumption of Requirement~\ref{req:GrantingDenying} holds. To show that the second assumption holds, we note that $p_j$ processes a request message $\messageForm{GET}$ from $p_i$ at a time $t_j \in [t_i,t_i+T_D]$ only when its delivery is timely (line~\ref{ln:MeassageAge}).  Thus, we only need to prove that a request sent by a requester $p_i$ at a time $t_i$ to a requestee $p_j$ in $p_i$'s membership $M$ is replied to with a $\messageForm{GRANT}$ if, and only if, at a time $t_j \in [t_i, t_i+T_D]$ when the request is delivered to $p_j$, $\noPriorityViolationDyn(AR_j) = \true$ and $p_j$ is in one of the states presented in the cases (i) to (iii) described in Requirement~\ref{req:GrantingDenying}.
    
    From line ~\ref{ln:GRANTDENYCondition} we observe that $\noPriorityViolationDyn(AR_j) = \true$ is a mandatory condition for a request to be granted. 
    From line~\ref{ln:GRANTDENYCondition}, we can also observe that if (a) $p_j$ has status $NORMAL$ or $TRYGET$, (b) $p_j$'s status is $GRANT$ and the granted agent is $p_i$, or (c) $p_j$'s status is $GET$ and $p_i$ has higher priority according to $\precedes$, then $p_j$ will send a $\messageForm{GRANT}$ and call $\ntfyEstimat(grant)$ (line ~\ref{ln:ntfyEstimatGrant}). Moreover, $p_j$ reply with a $\messageForm{DENY}$ if $p_j$ is in a state different from cases (a) to (c) or if $\noPriorityViolationDyn() = False$ (line~\ref{ln:SendDENY}).
    The cases (a) to (c) indeed correspond to (i) to (iii), which Requirement~\ref{req:GrantingDenying} defines. Since the lines \ref{ln:ntfyEstimatGrant} and \ref{ln:SendDENY} are the only locations in the algorithm where a $\messageForm{GRANT}$ or $\messageForm{DENY}$ is sent, Requirement \ref{req:GrantingDenying} is satisfied.

	\item \textbf{Requirement~\ref{req:explicitRelease}:~~} 
	
	From Algorithm \ref{alg:agent-events} we see that a $\messageForm{RELEASE}$ message is sent to all agents in $\D$ only on the lines \ref{ln:nexistsMtypeDENYelse}, \ref{ln:TRETRYRelease}, \ref{ln:statusgetsNORMAL}, and \ref{ln:ReleaseNtfyEstimatRevoke}. We need to show that these lines are reached if, and only if, any of the cases (i) to (iv) presented in Requirement \ref{req:explicitRelease} hold. We further need to show that $\ntfyEstimat(\revoke)$ is only called when an agent $p_j$ receives a $\messageForm{RELEASE}$ from another agent $p_i$ and $p_j$ is currently holding a grant for $p_i$ or when $p_j$ is in state $EXECUTE$ and has left the intersection.
	
	To reach and execute line \ref{ln:nexistsMtypeDENYelse}, the requesting agent $p_i$ (with status $GET$) needs to receive replies before the time $T = t_i + 2T_D - \TdoForeverLoop$, where $t_i$ is the time when $p_i$'s request was sent, otherwise the $\tRETRY$ timer will expire before the if-state condition on line \ref{ln:statusGETlast} evaluates to $\true$. The agent $p_i$'s $\tRETRY$ timer will be stopped immediately after $p_i$ discover that all needed replies have arrived to stop the timer from interfering (line \ref{ln:doForeverStopTRETRY}). Furthermore, any received $\messageForm{DENY}$ message before time $T$ from another agent $p_j$ will be stored by $p_i$ (line \ref{ln:RecieveDeny}). However, $p_j$'s $\messageForm{DENY}$ will only be kept if $p_j \in (\D \cap M$), where $M$ is $p_i$'s current membership  (line \ref{ln:removeStale}). The agent $p_i$ will then execute line \ref{ln:nexistsMtypeDENYelse} since $p_j$'s $\messageForm{DENY} \in \mathcal{M}$ (line \ref{ln:nexistsMtypeDENY}). This case correspond exactly to case (i).
	
	If not all messages arrive before time $T$, then the $\tRETRY$ timer set by $p_i$ when sending the request (line \ref{ln:TRETRY2TD}) expires and $p_i$ will execute line \ref{ln:timerRETRYexpires}. The agent $p_i$ will then reach and execute line \ref{ln:TRETRYRelease} since $p_i$ has status $GET$. This corresponds to case (ii). An agent cannot reach line \ref{ln:TRETRYRelease} in any other way since all other cases in which $p_i$ starts the $\tRETRY$ timer (lines \ref{ln:nexistsMtypeDENYelse}, \ref{ln:startTimerTRETRYTA}, and  \ref{ln:TRETRYRelease}) the agent will have status $TRYGET$.
	

	To execute line \ref{ln:statusgetsNORMAL}, agent $p_i$'s status has to be $EXECUTE$ and $p_i$ has to have left the intersection so that $hasLeftCriticalSection(p_i, M)$ evaluates to $\true$ (line \ref{ln:EXECUTEHasLeftCS}). This correspond to case (iii).
	
	
	The last line in which a $\messageForm{RELEASE}$ can be sent, line  \ref{ln:ReleaseNtfyEstimatRevoke}, can be reached only if an agent $p_i$ has status $GET$ when receiving a $\messageForm{GET}$ message $m$ from another agent $p_j$ and if $\neg \precedes(m) = \true$, and $\noPriorityViolationDyn(AR_j) = True$ (lines \ref{ln:MessageGET}, \ref{ln:GRANTDENYCondition}, \ref{ln:MessageGETGETTRYGET}, and \ref{ln:ReleaseWhenGET}). This corresponds to case (iv).
	
 	We consider the case in which agent $p_j$ has the status $GRANT$ and holding a grant for $p_i$ while receiving a $\messageForm{RELEASE}$ message from $p_i$. In this case, $p_j$ calls $\ntfyEstimat(revoke)$ (line \ref{ln:statusgetsNORMAL}). The only other line where \\$\ntfyEstimat(revoke)$ is called is line \ref{ln:revoke}. This line can only be reached if an agent has status $EXECUTE$ and if $\hasLeftCriticalSection(\aID)$ evaluates to $\true$ (line \ref{ln:EXECUTEHasLeftCS}). Requirement \ref{req:explicitRelease} is thus satisfied.

	\item \textbf{Requirement~\ref{req:AllowingDeferring}:~~} 
	%
	%
	Note that $p_i$ does not remove messages that are needed (such as in line~\ref{ln:removeStale}). By line~\ref{ln:statusGETlast} to~\ref{ln:nexistsMtypeDENYelse} and the  definition of the function $\last()$ in Figure~\ref{alg:defs}, agent $p_i$ does not change its status to $EXECUTE$ before it receives all the needed responses. This change to $EXECUTE$ depends on the absence of a response that includes a $\messageForm{DENY}$, i.e.,  the case of being fully-granted. Otherwise, $p_i$ changes $status$ to $TRYGET$, sends a release message to all agents in $\D$ and restarts the request process after $T_A$ time units.  
\end{enumerate}
\end{proof}

	%
	%

\begin{theorem}[Progression]
	\label{thm:MembershipService}
	Let $E$ be a synchronous execution of algorithms~\ref{alg:membership} and~\ref{alg:agent-events} in which admissible crossing opportunities occur infinitely often (for every direction of the intersection). Requirement~\ref{req:Progression} holds throughout $E$. 
\end{theorem}
\begin{proof}
Let $p_i$ be an agent that is the closest to the intersection from its direction and that aims at crossing the intersection, i.e., it calls $\tryManeuver()$ (line~\ref{ln:tryManeuverTurn}). By lines~\ref{ln:statusARMR} and~\ref{ln:statusGRANTstatusGRANTGET}, immediately after the execution of $\tryManeuver()$, the state of $p_i$ cannot be in $\{NORMAL, GRANT\}$. This proof shows that the states $GET$, $TRYGET$, and $GRANTGET$ must lead to $EXECUTE$ before returning to the state $NORMAL$
%
%
by demonstrating that $p_i$ indeed reaches the status $EXECUTE$ and crosses the intersection.

Without loss of generality, suppose that $E$ is an execution in which $p_i$'s status is not in $\{NORMAL$, $GRANT\}$ during the first system state of $E$. Moreover, suppose that immediately after $E$'s starting system state, we have a period of $T_{admissible}>B\_RETRY+ 2\TdoForeverLoop$ time units in which $p_i$ has an admissible crossing opportunity, where $B\_RETRY=\max\{2T_D, T_A\}$. Generality is not lost due to the statement in the beginning of this proof and because Requirement~\ref{req:Progression}'s assumptions, which means that an admissible crossing opportunity for $p_i$ always occurs eventually (within a finite period). 

The rest of the proof considers $p_i$'s status in $E$'s starting system state, which can be one of $GET$, $TRYGET$, and $GRANTGET$, and shows that it indeed reaches status $EXECUTE$.	

\begin{itemize}
	\item \textbf{Suppose that $\mathbf{status_i=GET}$.~~} 
	Note that the do forever loop (line~\ref{ln:doForever}) of Algorithm~\ref{alg:agent-events} starts a new iteration once in every $\TdoForeverLoop$ time units. Therefore, within $\TdoForeverLoop$ time units, agent $p_i$ executes line~\ref{ln:statusGETlast}. Once that occurs, the if-statement condition in line~\ref{ln:statusGETlast} holds, due to the assumption of this case that the status of $p_i$ is $=GET$, the definition of the function $\last()$ (Figure~\ref{alg:defs}) and the assumption that during the first $T_{admissible}$ time units of $E$ it holds that $p_i$ has an admissible crossing opportunity, i.e., $MR_i.\SM=\emptyset$ and $MR_i.\SM$ is  fresh.   
	
	Note that the removal of stale information in line~\ref{ln:removeStale} implies that $\nexists {m \in \M}:m.type = DENY$ holds. This means that the if-statement condition in line~\ref{ln:nexistsMtypeDENY} holds and $p_i$ gets status $EXECUTE$. Thus, the proof is done for this case, because  $T_{admissible}> 2\TdoForeverLoop$. 
	
	\item \textbf{Suppose that $\mathbf{status_i=TRYGET}$.~~} 
	By Corollary~\ref{thm:timerConsistency}, it holds that $status = TRYGET$ implies that the timer $\tRETRY$ is active. Thus, within $B\_RETRY=\max\{2T_D, T_A\}$ time units agent $p_i$ executes the procedure associated with this timer expiration (line~\ref{ln:timerRETRYexpires}), which calls the procedure $\tryManeuver(tag.turn)$ (line~\ref{ln:tryManeuvertagturn}). During the execution of the procedure (line \ref{ln:tryManeuverTurn}), it holds that the if-statement condition in line~\ref{ln:statusNORMALTRYGET} holds (due to this case's assumption that $status =TRYGET$). By the assumption that during the first $T_{admissible}$ time units there is an admissible crossing opportunity in $E$ with respect to agent $p_i$, it holds that the if-statement condition in line~\ref{ln:tsMRturntm2TMMRMO} evaluates to $\true$. Once that happens, line~\ref{ln:statusGET} assigns $GET$ to $status_i$ before setting the timer $\tRETRY$ to a value larger than $2\TdoForeverLoop$ (line~\ref{ln:TRETRY2TD}). Thus, $p_i$ start and completes the execution of the do forever loop (line~\ref{ln:doForever} to~\ref{ln:statusgetsNORMAL}) with $status_i=GET$ before the expiration of the timer $\tRETRY$ (line~\ref{ln:timerRETRYexpires}). The proof is thereby done for this case by arguments similar to the case of $status_i=GET$, since $T_{admissible}>B\_RETRY+ 2\TdoForeverLoop$. 
		
	\item \textbf{Suppose that $\mathbf{status_i=GRANTGET}$.~~} 
	By the definition of an admissible crossing opportunity, $p_i$'s membership will be empty for at least $T_{admissible}>B\_RETRY+ 2\TdoForeverLoop$ time units. An empty membership implies that no other agents can cause a priority violation with $p_i$, which in turn implies that the agent $p_j$ that $p_i$ is granting must have left the intersection. Within $\TdoForeverLoop$ time units, $p_i$ evaluates the if-statement in line \ref{ln:GRANTHasLeftCS} to $\true$, revokes the held grant, changes status to $TRYGET$, and calls $\tryManeuver$ (line \ref{ln:GRANTClearID} to \ref{ln:prevStatusGRANTGETtryManeuver}). The rest of the proof is identical to the latter part of the $TRYGET$ case above. 
	
\end{itemize}
\end{proof}

\chapter{Implementation}
\label{sec:implementation}

\Section{The Risk Estimator}
\label{sec:RE}

The \textbf{Risk Estimator} (RE) used in the simulations is based on~\cite{Lefevre2013IntentionAwareRE}, and was modified in order to utilise  \textit{speed models} -- as opposed to \textit{constant speed} -- when making vehicle motion predictions. 

The RE represents every vehicle in the traffic scene with a \emph{particle filter}. Each \textbf{particle filter} has weighted particles, where each particle has a position, heading, speed, manoeuvre intention, and expectation and intention to go or stop, representing the state of the vehicle. These particles are projected one time step forward in time based on their current \emph{state vector} using simple motion formulas (position, heading, speed) and a \emph{Hidden Markov Model} (HMM) (Expectation and Intentions). New weights are then assigned to the particles based on their similarity to the last received AS from the respective vehicle. A multivariate normal distribution is used as a likelihood function to measure similarity and the measured entities are assumed to be independent.

After the batch of particles for every vehicle have been projected and re-weighted, the risk estimator then estimates the new expectations.  Expectation is based on the traffic priority rules and time gaps of when the vehicles are estimated to enter the intersection, we refer to this time as \textbf{Time To Entry} (TTE). Speed models for every manoeuvre and intention to go or stop are used to estimate TTEs. For non-priority vehicles, the difference between the TTE for each vehicle is fed into a gap model, which determines if the gap is large enough for the vehicle to be expected to go, or whether it will be required to stop and wait for higher priority vehicles. The cumulative distribution function for the gap model is presented in Equation \eqref{eq:gap_model},
 \begin{equation}\label{eq:gap_model}
     Y=\frac{1}{1+e^{b(a-x)}}.
 \end{equation}
where the constants $a$ and $b$ can be adjusted to change the \textit{threshold position} and \textit{steepness} of which time gaps are considered big enough with a high likelihood. Priority vehicles will have the expectation to go.

As mentioned previously, the \textit{manoeuvre intention} and \textit{intention to go or stop} is determined by a HMM, see Figure \ref{fig:HMMIntention}. The HMM is designed let a vehicle change the manoeuvre intention with a probability of 10\,\% in every time step, representing the fact that most drivers keep to their manoeuvre plan. Drivers are also modelled as usually adhering to the expectations, with a chance of 90\,\% to comply if the current expectation matches the previous intention and a chance of 50\,\% if they do not match.

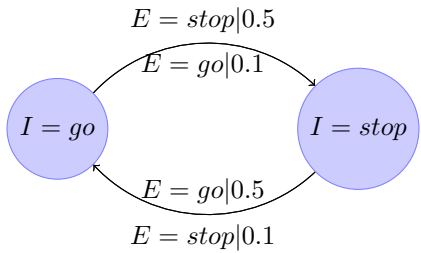
\begin{figure}[htbp]
\begin{center}
\begin{tikzpicture}[]
\node[state] (s1) at (0,4) {$I=go$};
\node[state] (s2) at (4,4) {$I=stop$}
    edge [<-,bend right=45] node[auto,swap] {$E=stop|0.5$} (s1)
    edge [<-,bend right=45] node[auto] {$E=go|0.1$} (s1)
    edge [->,bend left=45] node[auto,swap] {$E=go|0.5$} (s1)
    edge [->,bend left=45] node[auto] {$E=stop|0.1$} (s1);
\end{tikzpicture}
\end{center}
\caption{Visual representation of the HMM for the intention to go or stop. The transitions, labelled $a|b$, represent the likelihood $b$ of changing to the other value if the expectation is calculated to be $a$.}
\label{fig:HMMIntention}
\end{figure}

In the beginning of a new round of estimations, re-sampling will occur if there are too few particles with a high enough weight. Particles with a weight under a certain threshold will then be discarded and more likely particles will be duplicated. This process is called \emph{importance sampling} and is used to reduce the risk of \emph{sample depletion} when almost all particles have negligible weights, while still minimising the amount of lost information ~\cite{ParticleFilter}. If all weights equal $0$ after an estimation round, the particle filter is restarted with new particles spread around the latest received agent state.

Finally, the risk is computed for a vehicle by computing the sum of the normalised weights of the particles representing it where the intention is go but the expectation is to stop. A sum higher than the set risk threshold triggers an \textbf{emergency break} (EB) in which the vehicle sets its intention to stop. 

Having an intention to stop but expectation to go may also be considered risky when multiple vehicles are driving in a line but we decided to not include it in our calculations at this stage since we are studying collisions caused by vehicles coming from different directions.

\Section{Internal hierarchy for deciding vehicle intention}

Both the Manoeuvre Negotiation protocol and the  Risk Estimator influence a vehicle's intention to go or stop. According to the Manoeuvre Negotiation Protocol, all vehicles should have the intention to stop until they are fully granted. The Risk Estimator, on the other hand, sets a vehicle's intention based on a gap model to calculate the expectation and a HMM to, based on the previous intention and the current expectation, calculate the intention. The Risk Estimator also calculates a risk of a vehicle's intention not matching the expectation and triggers an emergency break if this risk passes a set threshold. Emergency break in this case is the same as setting the intention to stop. In addition to our safety system, traffic rules can also be used to set a vehicles intention. There are therefore 4 possible sources of information a vehicle can use to set its intention and setting up an internal hierarchy between them is necessary to avoid a arbitrary behaviour. 

The hierarchical order we decided upon with the most influential indicator first is
\begin{enumerate}
    \item Emergency break
    \item Priorities set by the Manoeuvre Negotiation Protocol
    \item Risk Estimator's HMM
    \item Priorities determined by traffic rules
\end{enumerate}
An emergency break will set the intention of a vehicle to stop no matter what the other system's say. If no emergency break is triggered then the priorities set by the Manoeuvre Negotiation Protocol decide if a vehicle should go or stop. Lastly, if no emergency break is triggered and the manoeuvre negotiation protocol is not in use, the output from the HMM decides the intention of a vehicle. 

If none of our safety systems are active then the traffic rules themselves can be used to set the intention. The traffic rules are placed last on the list but they influence the decisions made by both the Risk Estimator and the Manoeuvre Negotiation Protocol.

\Section{Architecture}
The architecture setup can be found in Figure \ref{fig:architecture}. An intersection traffic scenario, a membership service and two cars -- each following the manoeuvre negotiation protocol and having their own risk estimator -- were simulated using ROS. A storage server for vehicle state data and memberships was implemented as an Apache Zookeeper server.

\Subsection{Storage server}

The manoeuvre negotiation protocol relies on a storage server where agents regularly upload their current state and read their memberships computed by a membership service, here referred to as \emph{the cloud}. In our implementation, this storage server was chosen to be an Apache Zookeeper server because it provides a framework for building a highly reliable and distributed server. Communication to the server was done via TCP. 

The Zookeeper server has 2 main directories (called znodes): \emph{segment} and \emph{mr}. Segment stores the \emph{agent states} (time stamp, position, heading, speed) that all vehicles regularly update and that the cloud uses to compute memberships. The memberships come in sets, with one membership per possible manoeuvre through the intersection. Each membership is stored in \textit{mr} under the ID for every vehicle and has the form (Turn, Manoeuvre opportunity, Safety membership)\footnote{See Section \ref{sec:AlgDescMembership} for a more detailed explanation of the parts of a membership}. 

The vehicles periodically write their current state to the storage server and read their own membership set, whilst the cloud in turn reads the agents' agent states and then writes their computed memberships to the server (see Figure \ref{fig:architecture}).

\begin{figure}[ht!]
    \centering
    \includegraphics[width=\textwidth]{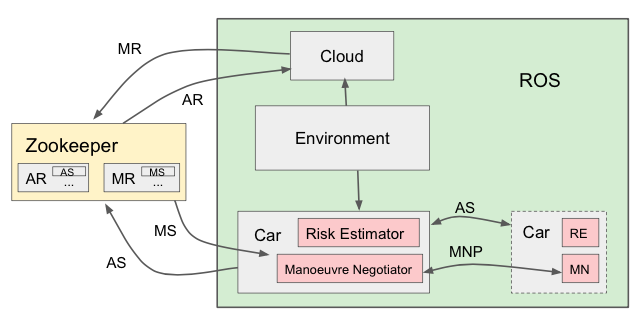}
    \caption{A high-level map of the simulation environment. One or two-headed arrows designate communication channels that are mono-respective bi-directional. Both the vehicles and the cloud were set up in ROS while a Zookeeper server was used as an external storage server. The environment block represents information regarding the geometry of the intersection.}
    \label{fig:architecture}
\end{figure}

\Subsection{Simulation environment}

Every vehicle is represented by a ROS node in the simulation environment and the vehicles publish their Agent State (AS) every $T_A$ time units to their own ROS topic (a publish-subscribe communication channel in ROS). The vehicles subscribe to other vehicle's topics and use the received states as inputs to their risk estimator. Every vehicle further sends their AS and retrieves their membership set from the storage server over TCP. The messages sent in the manoeuvre negotiation protocol uses UDP since the protocol itself is designed to handle packet loss and a much delayed (possibly retransmitted) packets are discarded in the protocol anyway.

The membership service is also implemented on a ROS node. It reads all the agent states (AR) from the \textit{segment} znode on the storage server and then writes all the computed memberships (MR) to another znode, \textit{mr}, on the same storage server. The memberships are computed using the ASs and intersection information, which includes the priorities and geometry set up at when the simulation was initiated.


\Subsection{The Risk Estimator}
\label{sec:RiskEstimator}

Every simulated vehicle was equipped with a \textbf{Risk Estimator}. The risk estimator chosen was based on one developed by by~\cite{Lefevre2013IntentionAwareRE} but which some modifications\footnote{see Section \ref{sec:RE}}.  The risk estimator was run as a background process throughout the simulation, as explained in further detail in Section \ref{sec:RE_Background}. Each particle filter used by the RE was set to have 625 weighted particles representing the state of a vehicle.

\Subsection{The Manoeuvre Negotiation Protocol}

Every vehicle node in ROS was equipped with an object to handle the Manoeuvre Negotiation Protocol in Algorithm \ref{alg:agent-events}. Another node, the \emph{Cloud}, was dedicated to running the Membership Algorithm in Algorithm \ref{alg:membership} in which memberships for every vehicle in the traffic scene are computed. 

More details on the Manoeuvre Negotiation Protocol are provided as line by line descriptions in Section \ref{sec:AlgorithmDescription}; a correctness proof is given in Section \ref{sec:corrProof}.

\chapter{Evaluation}
\label{sec:evaluation}

This chapter describes the different test scenarios, evaluation criteria and how tests were compared to validate our implementation of the risk estimator and the manoeuvre negotiation algorithm. Tests were also conducted to show liveliness of the updated protocol.

\Section{Simulation environment}

A 4-way intersection was set up in ROS with predefined paths for every incoming directions' manoeuvres, corresponding to turning left or right or going straight. The north-south lanes were given higher priority in respect to the crossing east-west lanes. The simulated vehicles were implemented on ROS nodes. A speed model decided the vehicles' speed depending on their distance from the intersection, see Figure \ref{fig:SpeedModel}. A \emph{PID-controller} adjusted the steering angle to keep the vehicles on the paths. The vehicles exchanged state information with added Gaussian noise every 0.5\,s using ROS-topics. The manoeuvre negotiation was conducted using UDP while the communication with the storage server used TCP. Every vehicle had their own instance of a risk estimator and manoeuvre negotiator.

\Subsection{Simulation scenario}

A great percentage of all traffic accidents with serious or fatal outcomes occur in intersections. Of all the possible turns through intersections, left-turn-across-path is considered one of the most collision-prone which is why we have decided to focus on this manoeuvre in our tests. The simulations involve two vehicles approaching an unsignaled intersection from opposite directions, where one vehicle \textbf{intends to do a left turn}, \vlow, while the other \textbf{intends to go straight}, \vhigh, see Figure \ref{fig:Test_scenarios_intersection}. Both vehicles are on the priority lane but they do not have equal priority in the intersection since intending to doing a left turn initially decreases a vehicle's priority. The test cases will therefore involve \vlow requesting permission from \vhigh to change the priorities.

\begin{figure}[!ht]
    \centering
    \includegraphics[width=0.8\textwidth/2]{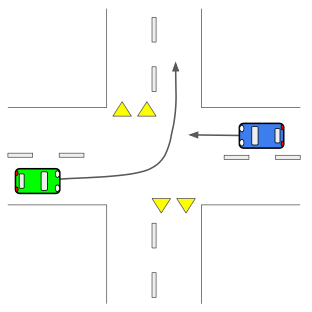}
    \caption{The scenario of a vehicle, \vlow, turning left across the priority lane with an opposing priority vehicle, \vhigh, approaching, used for testing our system.}
    \label{fig:Test_scenarios_intersection}
\end{figure}

\Subsection{System comparisons}

To evaluate our contribution to the Manoeuvre Negotiation protocol we want to compare its performance to tests conducted with the risk estimator alone and with the two systems combined. These tests are needed to check if the combined system performs better than any of the systems by itself. Values that will be compared are explained further in Section \ref{sec:Evaluation_criteria}. Reference values without emergency break and manoeuvre negotiation were calculated from the starting position in every case instead of simulating their runs since these vehicles deterministically following the speed model without stopping. The following setups were tested for every test case 
\begin{itemize}
    \item \textbf{RE}: Running the risk estimator alone with the emergency break feature on
    \item \textbf{MN}: Using the Manoeuvre Negotiation protocol alone, no emergency break
    \item \textbf{RE+MN}: Running the risk estimator with emergency break and the Manoeuvre Negotiator
\end{itemize}
Each of these setups were run with the same list of seeds for the random number generator to give the runs the same initial conditions. The seed affects the risk estimator -- which uses random numbers in the particle filter -- and the manoeuvre negotiation, which depends on a random function when a received request is processed. Every test case was run 10 times with 2 different seeds each time (one for each of the two libraries of random functions used). The quantitative evaluation criteria were computed by taking a mean of the 10 values with the highest and lowest value removed. The qualitative criteria were calculated as a mean over all 10 values, which results in a number between 0 to 1 indicating how likely an event is when certain start values are deployed.


\Section{Evaluation criteria}
\label{sec:Evaluation_criteria}
The tests were evaluated based on criteria that measure the:
\begin{itemize}
    \item Number of collisions and dangerous situations. A \textbf{collision} is here defined as a case where the vehicles both occupy the space shared by both their paths and a \textbf{situation} is defined as \textit{dangerous} if the distance separating the front of each vehicle is less than $4\,m$. Saved state data from the simulations was used to calculate these values.
    
    
    \item \textbf{Time Lost due to Priority Violations} (TLPV). TLPV is only computed for the initially higher priority vehicle \vhigh and will be positive that vehicle has to slow down due to the actions of the other (lower priority) vehicle \vlow. This would occur when, for example, either the emergency break is triggered or if the other vehicle is both granted and then releases the grant later than expected. TLPV is a qualitative measure of the implementation of $\noPriorityViolation$. Time to enter the intersection for each of the test cases were compared with a computed optimal value from a scenario where the priority vehicle does not have to interact with any other vehicle.
    
    \item Number of emergency breaks triggered by the risk estimator. The outcomes from the RE test case are used here to determine whether or not adding the Manoeuvre Negotiation Protocol will cause the risk estimator to give false alarms. 
    
    \item \textbf{Time To get fully Granted} (TTG) after the initial request was sent. This measure was recorded to demonstrate liveliness of the Manoeuvre Negotiation Protocol. It also provides a measure of throughput.
    
\end{itemize}
Emergency breaks and TLPV can both be signs of priority violations for the higher priority vehicle, caused by the risk estimator and Manoeuvre Negotiation Protocol separately. An emergency break signal will occur if the risk estimator computes that the lower priority vehicle is about to drive through the intersection before the higher priority vehicle when the time gap is considered too small. TLPV instead measures how much the priority vehicle is slowed down by having to break in order to let the granted lower priority vehicle finish its manoeuvre. It is therefore a measure of how well the manoeuvre negotiation protocol approximated the time of the manoeuvre execution for the lower priority vehicle. 

To measure the performance of our implementation of the risk estimator we further conducted some tests measuring the:
\begin{itemize}
    \item \textbf{Time To Collision Point} (TTCP) when the risk estimator first triggered an emergency break. This provides a measure of how early the risk estimator was able to detect a possible collision. The TTCP was computed using the vehicle's position in combination with the speed model shown in Figure \ref{fig:SpeedModel}, which produced an estimate as to when the vehicle would first reach the \emph{collision zone} (at the collision point). The collision zone is the part of the two vehicles' trajectories that overlap when also considering the size of the vehicles.
    
    \item Precision, which determines if the risk estimator is prone to producing a large number of false alarms. \textbf{Precision} is here defined as the number of test cases where an actual dangerous situation was detected, divided by the total number of the cases where the emergency break was triggered.
    
    \item Recall, which measures how well the risk estimation is able to detect potential collisions. \textbf{Recall} is here defined as the number of test cases where an emergency break was triggered in a dangerous situation divided by the total number of test cases where a dangerous situation occurred.
    \end{itemize}
    
Observe that TTCP is defined to be a non-negative measure, because any instance where an EB event is triggered and the vehicle is either in the collision zone of has left the collision zone are assigned a TTCP value of $0$, since at that point the vehicle has either reached or passed the collision point. It follows that any EB events after this point cannot help in preventing a collision.




\Section{Test Cases}
\label{sec:test_cases}



The distance parameters used during the simulation are all measured from the centre of the intersection, see Figure \ref{fig:parameter_expl}. The start distance for \vlow, $d_0$, and the request initiation distance, $d_{Init}$, were kept fixed to $65\,m$ and $30\,m$ respectively. The start distance for \vhigh, $d_1$, was varied \footnote{see Section \ref{sec:caseNormal} for more details}.

\begin{figure}[!ht]
    \centering
    \includegraphics[width=0.8\textwidth]{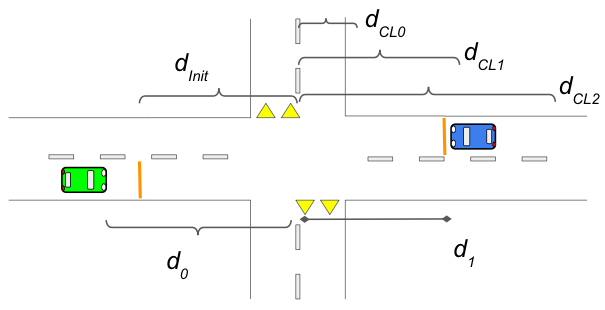}
    \caption{Distance parameters used during the simulations. All distances are measured from the middle of the intersection and all except $d_1$ are fixed. The values for $d_{CLX}$ can be found in Table \ref{tab:Comstop}, $d_{Init}=30\,m$, $d_0=65\,m$ and the values for $d_1$ are described in \ref{sec:test_cases}.}
    \label{fig:parameter_expl}
\end{figure}

Each test case focused on varying one parameter while keeping the others fixed or, as in the cases with added noise or communication loss, limited to a few values in order to try to determine how varying the parameter affects the results. Here follows a presentation on the different test cases and a motivation as to why they where chosen to evaluate our implementation. The letter combination used to refer to each test case is set between parenthesis in the test case's title.

\Subsection{Test Case: Normal (Normal)}
\label{sec:caseNormal}

Variations in starting distance result in different time gaps of when the two vehicles are estimated to enter the intersection and should therefore affect everything that is related to our risk estimation. The start distance $d_1$ for \vhigh was therefore varied in steps of 4\,m between 125 to 13\,m, measured from the centre of the intersection. The other vehicle's start distance $d_0$ was kept constant. 




\Subsection{Test Case: Noisy Measurements (Noise)}
\label{sec:caseNoise}

Noise in various levels is always present in real-life measurements. It is therefore of great importance to test how our implementation performs under variously noisy conditions. 

In this test case, extra noise was added to the measured agent state values including speed, position and angle before using the values in the system. Since the noise added in the Normal case in Section \ref{sec:caseNormal} was Gaussian with a true mean and standard deviation of $\sigma$, we decided to add two test cases here, one with a standard deviation for the noise of $1.5\sigma$ and the other of $2\sigma$. For each of these two settings, the different start position runs as described in Section \ref{sec:caseNormal} were performed.


\Subsection{Test Case: Lost communication (ComLoss)}
\label{sec:caseLostCom}

Message loss is also a common fault that occurs in real systems. To simulate message loss in our simulation we added a total lack of communication for \vlow at 3 different distances from the intersection $d_{CL0}$, $d_{CL1}$, and $d_{CL2}$, for 3 different time durations $t_{CL0}$, $t_{CL1}$, and $t_{CL2}$. Their values can be found in Table \ref{tab:Comstop} and the $d_{CLi}$-distances are depicted in Figure \ref{fig:parameter_expl}. Notice that the three starting values correspond to long before the request line, just before the request line, and just before entering the intersection. For every combination of $d_{comstop}$ and $t_{comstop}$, varied start position for \vhigh was applied as in Section \ref{sec:caseNormal}. 

\begin{table}[h!]
    \centering
    \begin{tabular}{c|c|c|c}
        $d_{CL}$ (m) & 11 & 31 & 51 \\ \hline
        $t_{comstop}$ (s) & 1 & 2 & 3
    \end{tabular}
    \caption{Values for the communication loss parameters $d_{comstop}$ -- determining how far from the intersection communication was cut for \vlow -- and $t_{comstop}$ -- how long communication was cut. These values were combined to form 9 different test cases.}
    \label{tab:Comstop}
\end{table}

This test should give an indication on how well the systems are able to handle a lack of new information for a limited time period at different distances from the intersection. 


\Subsection{Test Case: Traffic Offenders (Offender)}
\label{sec:caseOffender}

Not all vehicles follow traffic rules and conventions, so a robust system has to prevent traffic accidents even if traffic offenders are present. A traffic offender was simulated by letting \vlow ignore priorities while still sending positional messages and other messages required by the protocol. The offender has its intention set to $go$ throughout the simulation, which implies that the vehicle drives without following the expectation set by the risk estimator or priorities set by the protocol, and does not respond to triggered emergency breaks. The test runs had the same starting conditions as in Section \ref{sec:caseNormal} with varying start distance $d_1$.
The purpose of the test is to evaluate the risk estimator in an environment where not all vehicles are adhering to their safety system, unlike all previous test cases presented here.

\Section{Evaluation terminology}
\label{sec:Evaluation_terms}
We have developed the following terminology for ease of reading:
\begin{itemize}
    \item A particular \textbf{Setup} will be denoted $\mathcal{S}\,_{Type}\,$. E.G. The \textit{RE} setup will be denoted \sre.
    \item A particular \textbf{Test Case} will be denoted $\mathcal{TC}\,_{Type}\,$. E.G. The \textit{Normal} test case will be denoted \tcno.
    \item An \textbf{Evaluation Scenario} will refer to a combination of one test case $\mathcal{TC}$ and one setup $\mathcal{S}$, and will be denoted $\mathcal{ES}\,_{TestCase}^{Setup}\,$. E.G. \tcno paired with \sremn will be denoted \tcremn.
\end{itemize}

\chapter{Results}
\label{sec:results}

This section presents the results obtained from running the test cases described in Section 
\ref{sec:test_cases} and measuring the evaluation criteria defined in Section \ref{sec:Evaluation_criteria}. 

Our initial hypothesis was that \sremn would prove to be a more safe system than \sre, which was confirmed by the results: there were fewer collisions and dangerous situations recorded when using \sremn. This system also reduced the amount of triggered emergency breaks compared to \sre, while still keeping TLPV low. The Risk Estimator managed to detect all dangerous situations in \sre but was prone to give false alarms in some test cases. 

The Risk Estimator's performance in \sremn is hard to evaluate because of the small number of  test runs that resulted in a dangerous situation. Unfortunately, it is possible that characteristics of the simulation environment itself may have contributed to some  results, since no explanation based on characteristics of the safety system is available in some cases. For example, the simulations take some time to start up properly which can affect the cases where one of the vehicles starts close to the intersection.

A variable parameter common to each test case is the \textbf{start distance} for \vhigh, which we denote $d_1$. It represents the \textit{initial spatial separation} between the vehicles when compared to the fixed parameter $d_0$. Instead of plotting the measured entities against $d_1$, another entity is introduced to represent the \textit{initial temporal separation}, $\Delta t$. $\Delta t$ is defined as the \textit{estimated difference in time between when \vlow and \vhigh will each reach the intersection} based on their initial positions $d_0$ and $d_1$, and the presented speed model for ``go'' in Figure \ref{fig:SpeedModel}.


\Section{Request Time}
\label{sec:Results_RT}

\textbf{Time to Grant} ($TTG$) was calculated as the time difference between \textit{when a vehicle started its first request round} and \textit{when the request was fully granted}. We expect low \TTG values for \vhigh (denoted \vhighttg) since it has the highest priority and thereby should not have to wait for $\messageForm{GRANT}$ from \vlow. For \vlow (denoted \vlowttg) we expect a shark-fin pattern in the \TTG-$\Delta t$ plot, with low values for the lowest $\Delta t$, a sudden jump in \TTG when $\Delta t$ is low enough for \vlow to have to wait for \vhigh before crossing the intersection and then a steady decrease in \TTG as \vlow has to wait shorter and shorter times.

Recorded \vlowttg and \vhighttg for \tcnremn and \tccremn can be found in Figure \ref{fig:RT_case4}. The added noise does not seem to affect \TTG if compared to \tcno used for reference, denoted $N:0$, since the plotted points almost coincide for every value of $\Delta t$. \vhighttg is close to $0$ due to its empty membership, implying that it does not have to send out any requests. \vlowttg starts higher than \vhighttg since \vlow reaches the intersection first when $\Delta t <0$, and thus has to wait for a reply from \vhigh.  A slight rise in \vlowttg is visible after $\Delta t=-3$ and the plotted points then drop of linearly until they reach a value close to $0$. This drop corresponds to \vhigh estimating that the time gap between them is now too small for \vlow to safely perform its manoeuvre without \vhigh slowing down. \vlow's requests are thereby rejected until \vhigh leaves the intersection. Increased $\Delta t$ implies that \vhigh starts closer and closer to the intersection and \vlow thereby has to wait for shorter and shorter time periods. A slight periodicity in \vhighttg is distinguishable in the plot, with peaks at $\Delta t = -3, -1$, and $1$. It is unclear what is causing these small periodic changes in \TTG.

Similar patterns are present in the communication loss plots to the right in Figure \ref{fig:RT_case4}. Additional features in these plots includes an initially higher value for \vlowttg, temporary constant \TTG instead of declining linear behaviour and initially declining linear behaviour. The first case occurs for CLD:0, CLP:2, the second case occurs for CLD:2, CLP:2 and CLD:0, CLP:2, and the third case for CLD:1, CLP:2. Case 1 and 3 are probably caused by communication loss during the request rounds. For example, in the third case \vlow is trying to start a request round when its communication is switched off and it stays off for long enough for the vehicles to move close enough to produce a too small time gap. 

A slight increase in \vhighttg occurs after $\Delta t=3$. This is an effect caused by \vhigh starting in a position when it has already passed the request line, $d_1<d_{Init}$, and has to wait for the membership service to update its membership since a valid membership is needed to start a request round. This effect is thereby most probably caused by the slow upstart of the simulation environment and not a feature of the system.

\smn resulted in plots similar to \sremn and the respective plots have been omitted. The results for \TTG thereby live up to our expectations.


\begin{figure}[!ht]
\begin{minipage}{0.55\textwidth}
\hspace*{-0.15in}
\includegraphics[width=0.85\textwidth]{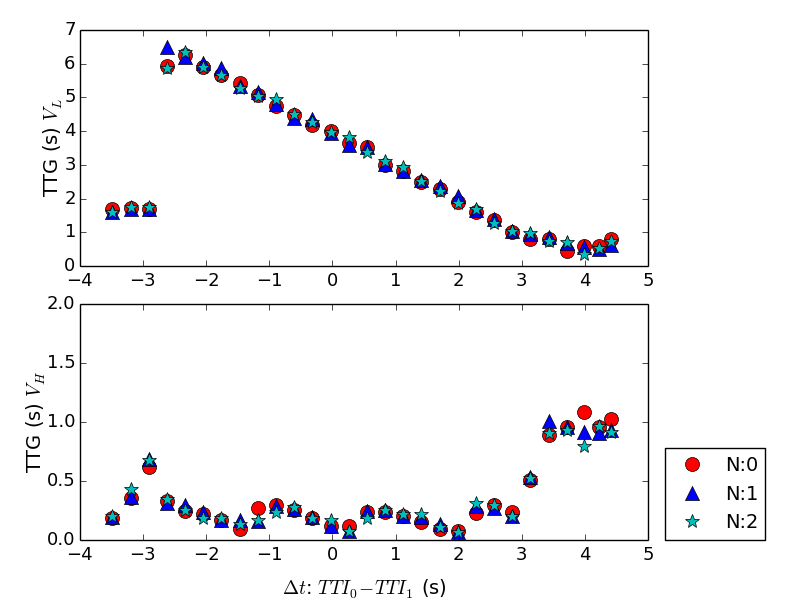}
\end{minipage}%
\begin{minipage}{0.65\textwidth}
\hspace*{-0.75in}
\includegraphics[width=0.85\textwidth]{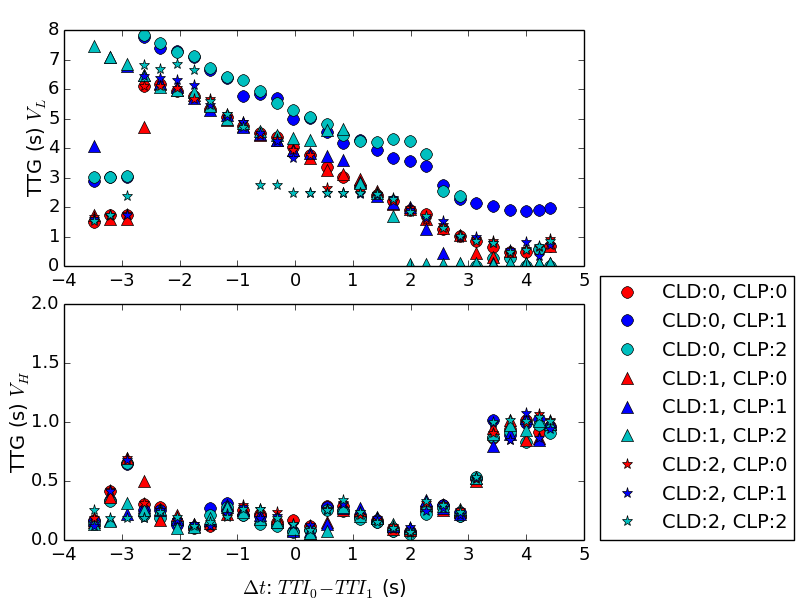}
\end{minipage}
\caption{Time To Grant (\TTG) for \tcnremn (left) and \tccremn (right). $V$ stands for vehicle and $N$ for noise level in the legend, where $N:0$ corresponds to the normal reference setup. \vhighttg -- where \vhigh has higher priority -- is fairly constant and keeps low while \vlowttg approximates a sharkfin pattern. The sudden increase in \TTG for both vehicles near $\Delta t=-3$ indicates a threshold value, after which \vhigh estimates that the gap between the two vehicles is too small and -- as a result -- it should not grant \vlow. }
\label{fig:RT_case4}
\end{figure}


\Section{Priority Violation Effects}
\label{sec:Results_TLPV}

\textbf{Time Lost due to Priority Violation} (TLPV) was estimated for \vhigh by taking the difference between the actual Time To Intersection (TTI) with the calculated ideal TTI when \vhigh follows the speed model with the intention ``go''. We expect values close to 0 for \smn and \sremn for the test cases \tcno and \tcnno since the Manoeuvre Negotiation protocol should preserve the initial priorities. In \tcono and \tccno, however, we expect to see increased TLPV due to emergency break (EB) and \vlow's lost messages.

The measured TLPV for \tcnremn and \tccremn  can be found in in Figure \ref{fig:TLPV_case4}. The different noise levels tested do not seem to have a visible effect on TLPV when compared to the normal case, denoted \textbf{N:0}. The shape of the normal curve can also be distinguished in the result plot for the \tccno when the results for varied \textbf{Communication Loss Period} (CLP) and the point of started communication loss, defined as \textbf{Communication Loss Distance} (CLD), are plotted together. Here we see a clear impact of both CLP and CLD on TLPV. The lower $\Delta t$ values, corresponding to \vlow starting closer to the intersection, are affected by $CLD:2$, which corresponds to communication loss for \vlow just before entering the intersection. For these cases, \vhigh has to wait for \vlow to start communicating again before its request can be granted. In the same way, CLD:1 gives rise to upgoing trends in the middle of the plot and CLD:0 a bit further right. The different slopes for the same CLD is caused by different CLP values where the highest CLP, corresponding to CLP:2, gives the steepest slope. 

\begin{figure}[!ht]
\begin{minipage}{0.55\textwidth}
\hspace*{-0.15in}
\includegraphics[width=0.85\textwidth]{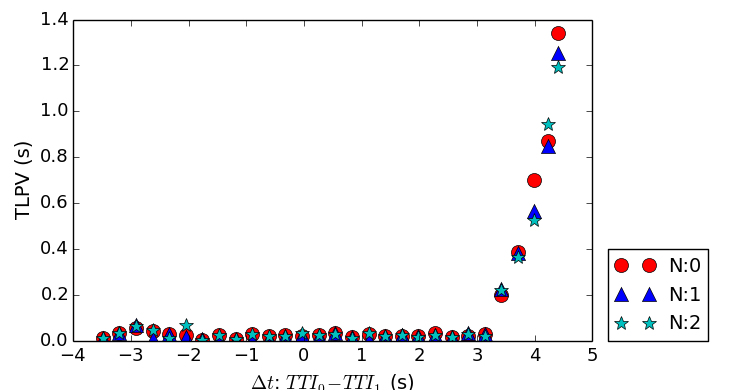}
\end{minipage}%
\begin{minipage}{0.65\textwidth}
\hspace*{-0.75in}
\includegraphics[width=0.85\textwidth]{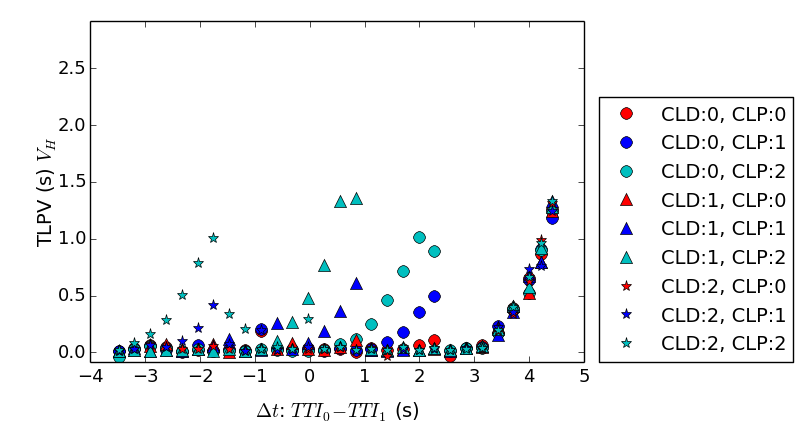}
\end{minipage}
\caption{TLPV for \tcnremn (left) and \tccremn (right). In \tcnremn TLPV is kept below 0.1\,s up for $\Delta t<3$. A steady rise after $\Delta t=3$ is present in both plots and is caused by a startup delay in the simulation. Additional rising trends for different CLD and CLP values are present in \tccremn.}
\label{fig:TLPV_case4}
\end{figure}

Figure \ref{fig:offenderTLPV_c4} displays the measured TLPV for \tcoremn. The TLPV has a peak just before $\Delta t=0$ which is caused by \vhigh's RE triggering an EB due to \vlow's unexpected behaviour.

\begin{figure}[!ht]
    \centering
    \includegraphics[width=0.5\textwidth]{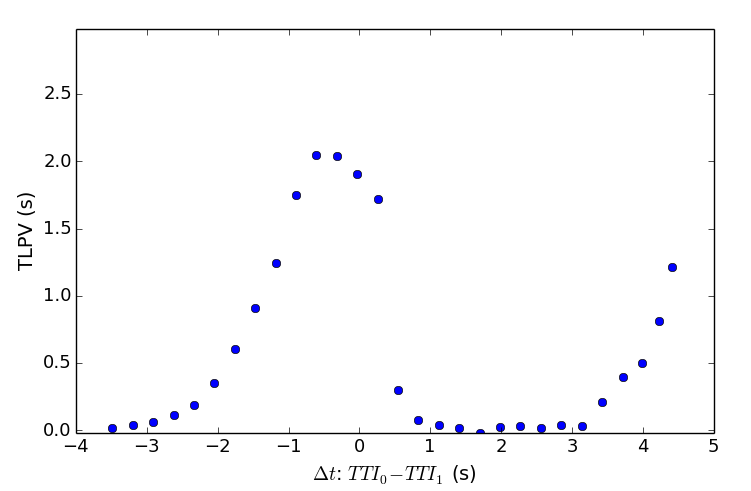}
    \caption{Time lost due to priority violation for \vhigh for \tcoremn. Compared to the \tcnremn in the left plot in Figure~\ref{fig:TLPV_case4} this plot has an additional peak just before $\Delta t=0$, caused by EBs for \vhigh when \vlow disobeys priorities and drives out in front of \vhigh.}
    \label{fig:offenderTLPV_c4}
\end{figure}

The results of test cases in \sre are almost identical to \sremn and their plots are therefore omitted. There is one exception, \tcomn, which is similar to \tcremn since \vhigh is not equipped with a risk estimator (RE) in this test case. 

The TLPV increases rapidly after about $\Delta t=3$ for all setups and test cases. This is caused by the same issue as the one explained in Section \ref{sec:Results_RT}, which causes \vhigh to wait for a valid membership. Waiting implies that the vehicle follows the ``stop'' speed model and therefore drives slower or even stops and waits just before the intersection causing a rise in TLPV. Apart from this, the results here matches our expectations with low TLPV in \tcno and \tcnno and some higher, but still -- for safety reasons -- agreeable, values for \tccno and \tcono. 


\Section{Emergency Break Events}
An Emergency Break event happens if the Risk estimator detects a risk over a set threshold. We expect none or just a few emergency breaks in \sremn since the vehicles actively negotiate their internal priority order and the Risk Estimator is notified of any priority change. In \sre we expect the number of EB triggered by \vhigh to be higher than the number of EB triggered by \vlow since \vhigh in this setup is expected to go throughout the simulation and thereby cannot have an expectation to stop and intention to go -- which is what triggers an EB.

The mean number of emergency break events for the \tcnre and \tcnremn scenarios is plotted in Figure \ref{fig:EB_noise_case2_4}. In \sre, increased noise seems to reduce the number of emergency breaks. This may be due to less clear intention caused by lack of new information, leading to a more evened out distribution of estimated intentions and expectations and less mismatch between the two. \vlow shows a low number of emergency breaks which is expected since it is the lower priority vehicle and in an ideal case should not rise any alarms when the higher priority vehicle approaches the intersection. That emergency breaks are at all recorded for \vlow may be caused by random functions used in the risk estimator and its gap model or due to randomness in the simulation.

In \tcnremn, just a few emergency breaks were recorded and out of these no mean exceeded 1 for any $\Delta t$. This implies that on average less than 1 emergency break per run was recorded per measurement set of 10 runs. The reason why almost no emergency breaks were triggered is because the manoeuvre negotiation protocol used in \sremn enables the vehicles to be aware of the other vehicle's intention and also informs the risk estimator of any priority change. In a best case scenario we would expect no emergency breaks but the few recorded is probably again due to randomness connected to the risk estimator or the simulation.



\begin{figure}[!ht]
\begin{minipage}{0.55\textwidth}
\hspace*{-0.15in}
\includegraphics[width=0.85\textwidth]{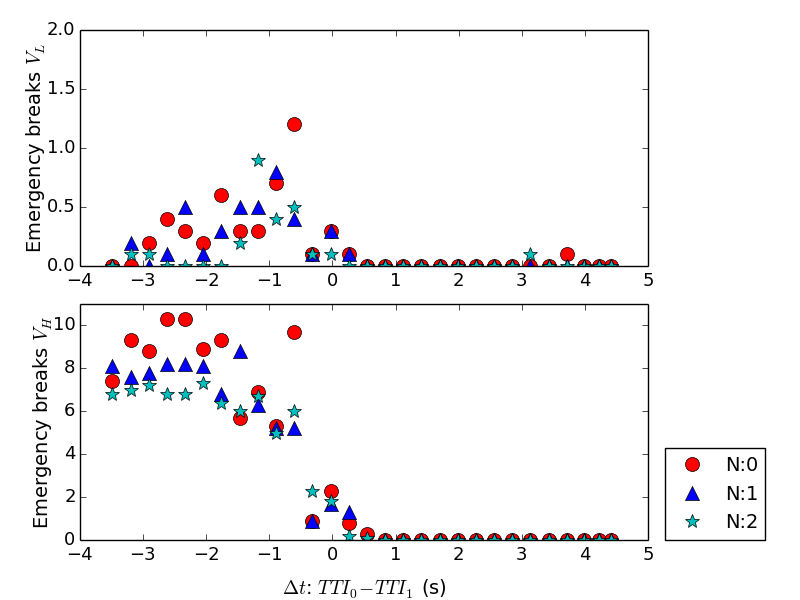}
\end{minipage}%
\begin{minipage}{0.65\textwidth}
\hspace*{-0.75in}
\includegraphics[width=0.85\textwidth]{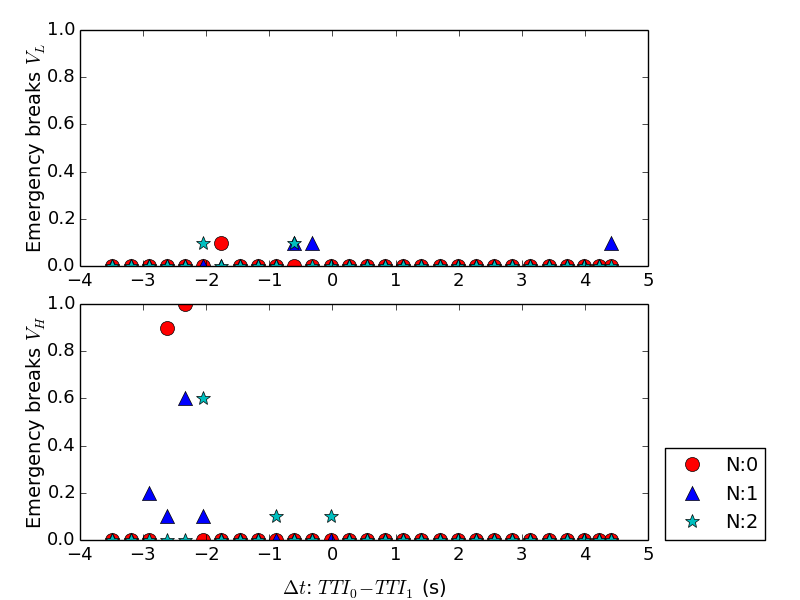}
\end{minipage}
\caption{Mean number of emergency breaks for \tcno (N:0), and \tcnno (N:x) for \vlow and \vhigh seen in the legend as V:0 and V:1. The left plot is the result from \sre and the right plot is from \sremn. Higher noise levels seem to lead to less number of emergency breaks. In \sremn just a few emergency breaks were triggered. Notice that $0.1$ implies $1$ triggered emergency break in $10$ runs.}
\label{fig:EB_noise_case2_4}
\end{figure}

In \tccre a drastic drop in emergency breaks by \vhigh can be seen when \vlow stops communicating in the intersection, labelled as CLD:2 in the left plot in Figure \ref{fig:EB_comloss_case2_4}. The largest drop is achieved for the longest communication loss period, denoted CLP:2. The result from \tccremn is displayed to the right in Figure \ref{fig:EB_comloss_case2_4}. Here, we see an increase in the mean number of emergency breaks when compared to \tcremn in Figure \ref{fig:EB_noise_case2_4}. However, most data points still lies beneath 1 and much lower than \sre.


\begin{figure}[!ht]
\begin{minipage}{0.65\textwidth}
\hspace*{-0.15in}
\includegraphics[width=0.85\textwidth]{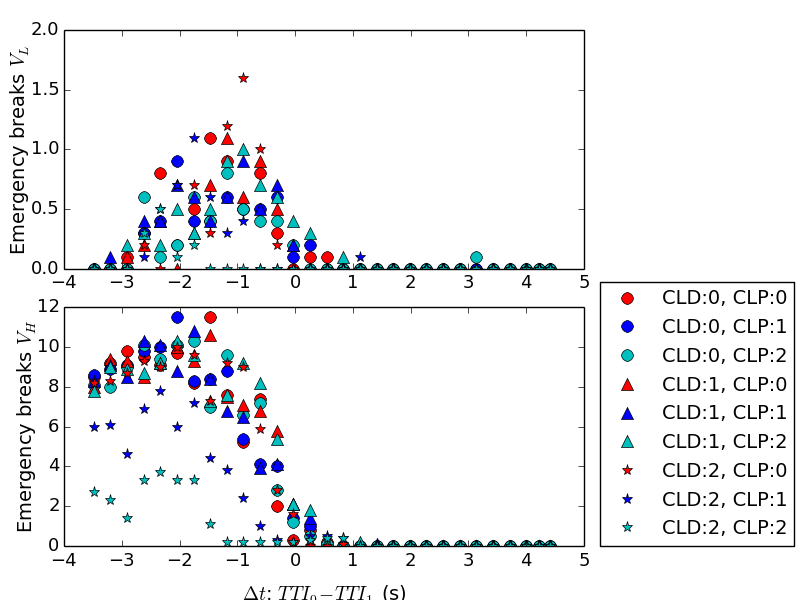}
\end{minipage}%
\begin{minipage}{0.55\textwidth}
\hspace*{-0.75in}
\includegraphics[width=0.85\textwidth]{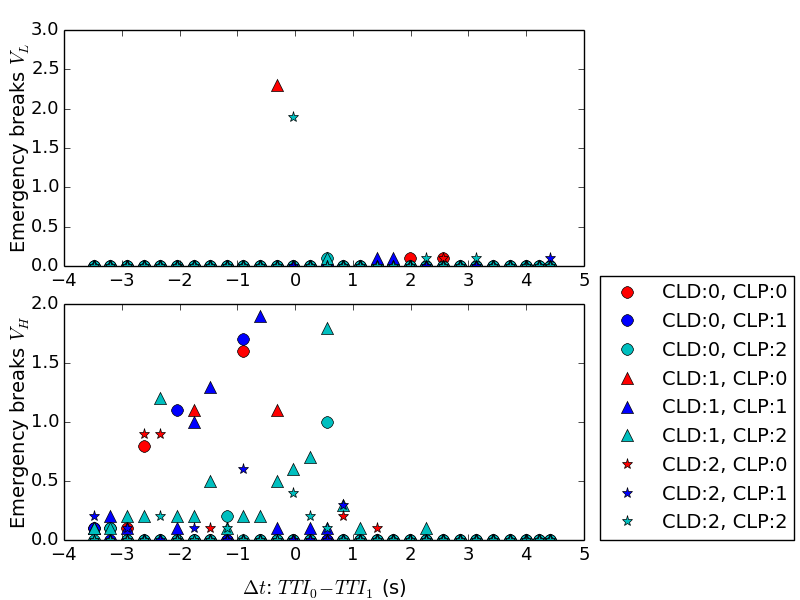}
\end{minipage}
\caption{Mean number of emergency breaks for the \tccre, left, and \tccremn, right, for \vlow and \vhigh for different communication loss periods (CLP) with varying start of the communication loss, CLD. The number of emergency breaks in \tccre is drastically reduced when \vlow stops communicating in the intersection, corresponding to CLD:2. More emergency breaks were recorded in \tccre compared to \tcre, as seen in the left plot of Figure \ref{fig:EB_noise_case2_4}.}
\label{fig:EB_comloss_case2_4}
\end{figure}

\Section{TTCP Upon First Emergency Break}
\label{sec:Results_TTC_EB}

The \textbf{Time To Collision Point} (TTCP) for a vehicle is calculated when the first emergency break (EB) is triggered to give a measure of how far in advance a potential can be detected. We expect that if any EB are triggered then they will occur when the other vehicle reach a certain distance from the intersection since the vehicles strictly follow the speed models. This would result in a declining TTCP since it is measured for the vehicle that triggered the EB.

The results for \tcre, \tcnre and \tccre is plotted in Figure \ref{fig:TTCP_case2}. All three test cases show a similar linear decline of about -1 for \vhigh when $\Delta t<0$ which indicates that the EB is triggered when \vlow is at approximately the same distance from the intersection in each case, approximately 1\,s away from reaching the intersection. A decrease in TTCP is visible for \tccno for increasing communication loss time (CLP) for \vhigh with $CLD:2$, corresponding to communication loss for \vlow when entering the intersection. Both plots show a few false alarms for $\Delta t>0$ for \vlow, and a lot of late alarms (TTCP$<0$) for $\Delta t<0$. \tcono is similar to \tcno and \sremn is similar to \sre but with less data points, so these plots are omitted.

The linear decline of TTCP for \vhigh matches our expectation that the detection would occur when \vhigh reaches a certain distance from the intersection. The late EB triggered by \vlow was not expected and may have been caused by unpredictability or implementation flaw of the risk estimator. 


\begin{figure}[!ht]
\begin{minipage}{0.55\textwidth}
\hspace*{-0.15in}
\includegraphics[width=0.85\textwidth]{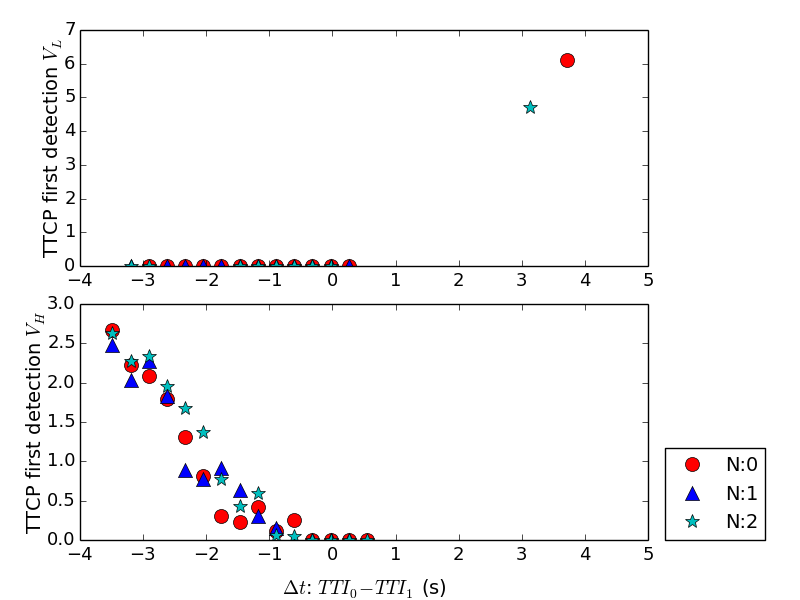}
\end{minipage}%
\begin{minipage}{0.65\textwidth}
\hspace*{-0.75in}
\includegraphics[width=0.85\textwidth]{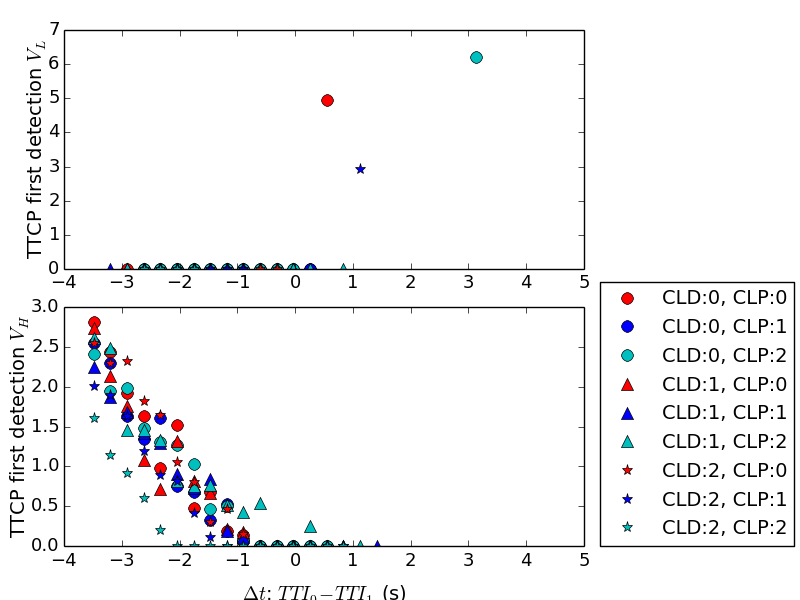}
\end{minipage}
\caption{Time To Collision Point (TTCP) for the earliest EB for \tcre and \tcnre (left) and \tccre (right). Both figures show a linearly declining TTCP with increasing $\Delta t<0$ for \vhigh and a few false alarms (where $\Delta t>0$).}
\label{fig:TTCP_case2}
\end{figure}

\Section{Collisions and Dangerous Situations}

Collisions between the two vehicles was calculated as when both vehicles at the same time were inside the area common to their trajectories if considering the vehicle width. A dangerous situation occurred if their fronts where closer than 4\,m apart, which implies that collisions are a subset of dangerous situations. We expect to see no collisions and at most a few dangerous situations in \sremn and \smn except for in \tcono where \vhigh may not have enough time to detect and stop for the offender (in \sremn). Collisions may on the other hand be present in \sre since the Risk Estimator has a best effort approach to safety.

The number of recorded collisions and dangerous situations for different Evaluation Scenarios can be found in Table \ref{tab:CollisionDanger}. \smn and \sremn yielded lower values than \sre in all cases. \smn resulted in less collisions than \sremn while \sremn resulted in less dangerous situations than \smn. The highest number of collisions and dangerous situations for \sre seems to occur when a communication loss occurs just before \vlow reaches the request line. 

The collisions in the \sremn scenario occurred in \tcoremn, which was expected, but also in \tccremn, which was not expected. Either, the longest period of communication loss for \vlow in \tccremn prevented both the Risk Estimator and Manoeuvre Negotiation protocol from keeping the vehicle safe in this one case, or this collision was caused by an undetected error in during the simulations. \sre resulted in a significantly higher amount of collisions and dangerous situations as expected.


\begin{table}[!ht]
    \centering
    \begin{tabular}{c|c|c|c|c|c|c}
         & Collision-RE & Collision-MN & Collision-RE+MN & Danger-RE & Danger-MN & Danger-RE+MN\\
         \hline
        Normal & 5 & 0 & 0 & 9 & 0 & 0 \\
        \hline
        \hline
        N0 & 5 & 0 & 0 & 8 & 1 & 0 \\
        \hline
        N1 & 1 & 0 & 0 & 8 & 0 & 0\\
        \hline
        \hline
        O & 2 & 1 & 2 & 4 & 4 & 3\\
        \hline
        \hline
        CL0-0 & 5 & 0 & 0 & 8 & 1 & 0\\
        \hline
        CL0-1 & 3 & 0 & 0 & 9 & 1 & 0\\
        \hline
        CL0-2 & 4 & 0 & 0 & 8 & 1 & 0\\
        \hline
        CL1-0 & 2 & 0 & 0 & 9 & 0 & 0\\
        \hline
        CL1-1 & 4 & 0 & 0 & 10 & 0 & 0\\
        \hline
        CL1-2 & 7 & 0 & 0 & 9 & 2 & 0\\
        \hline
        CL2-0 & 3 & 0 & 0 & 7 & 2 & 0\\
        \hline
        CL2-1 & 4 & 0 & 0 & 9 & 0 & 0\\
        \hline
        CL2-2 & 2 & 0 & 1 & 8 & 1 & 1\\
    \end{tabular}
    \caption{Number of test cases that resulted in a collision or a dangerous situation in \sre, \smn, and \sremn for the \tcno (N0), \tcnno (Nx), \tcono (O), and \tccno (CLX-Y). The maximum value is 29 corresponding to at least one collision or dangerous situation in all cases with different start distance $d_1$. \smn and \sremn resulted in much less collisions and dangerous situations than \sre.}
    \label{tab:CollisionDanger}
\end{table}


\Section{Risk Estimator Recall and Precision}

For the Risk Estimator's Emergency Breaks, recall was evaluated as total number of tests where an EB was triggered during a dangerous situation divided by the total number of cases where a dangerous situation occurred. Precision was calculated as the number of true dangerous situations that were detected and triggered an EB divided by the total number of cases were an EB was triggered. We expect a high Recall since the vehicles follow the same motion patterns as the Risk Estimator uses to infer their manoeuvre intention.

Recall and precision values for \sre and \sremn are presented in Table \ref{tab:RecallPrecision}. \vhigh has a recall value of 1.0 in \sre for all test cases, which implies that it was able to detect all dangerous situations despite increased noise, communication loss, and the other vehicle ignoring priorities. However, the lower precision results implies that the Risk Estimator is prone to rise false alarms. Safety is prioritised over false alarms by the Risk Estimator when run by itself. The result for \sremn is less clear since only a few dangerous situations were recorded for this setup, see Table \ref{tab:CollisionDanger}. 

In \sre we got high recall values as expected while more tests will have to be run in order to get a more reliable values for \sremn.

\begin{table}[!ht] 
    \centering
    \begin{tabular}{c|c|c|c|c}
         & Recall-RE & Recall-RE+MN & Precision-RE & Precision-RE+MN  \\
         \hline
        Normal & 0.56/1.0 & -/- & 0.28/0.50 & -/- \\
        \hline
        \hline
        N0 & 0.44/1.0 & -/- & 0.22/0.50 & 0.0/0.0 \\
        \hline
        N1  & 0.38/1.0 & -/- & 0.18/0.47 & -/- \\
        \hline
        \hline
        O & 1.0/1.0 & 0.67/1.0 & 0.19/0.19 & 0.10/0.15\\
        \hline
        \hline
        CL0-0 & 1.0/1.0 & -/- & 0.4/0.4 & -/- \\
        \hline
        CL0-1 & 1.0/1.0 & -/- & 0.38/0.38 & -/- \\
        \hline
        CL0-2 & 0.75/1.0 & -/- & 0.29/0.38 & 0.0/0.0 \\
        \hline
        CL1-0 & 0.78/1.0 & -/- & 0.32/0.41 & -/- \\
        \hline
        CL1-1 & 0.50/1.0 & -/- & 0.24/0.48 & 0.0/0.0 \\
        \hline
        CL1-2 & 0.89/1.0 & -/- & 0.33/0.38 & 0.0/0.0 \\
        \hline
        CL2-0 & 1.0/1.0 & -/- & 0.33/0.33 & 0.0/0.0 \\
        \hline
        CL2-1 & 0.44/1.0 & -/- & 0.21/0.47 & 0.0/0.0 \\
        \hline
        CL2-2 & 0.25/1.0 & 1.0/0.0 & 0.12/0.47 & 0.5/0.0 \\
    \end{tabular}
    \caption{Recall and precision results for \sre and \sremn where values separated with a ``/'' are the values for \vlow and \vhigh respectively. A ``-'' signifies that no value could be computed due to a lack of dangerous situations or emergency breaks. The Risk Estimator performs better on Recall compared to Precision and \vhigh detects all dangerous situations in \sre. }
    \label{tab:RecallPrecision}
\end{table}

\chapter{Discussion}
\label{sec:discussion}


\Section{Preservation of the Original Traffic Rules}

Traffic rules and conventions are designed not only to make driving more safe but also to promote fairness and a high traffic flow. It is therefore important that our combined safety system tries to adhere to the traffic rules and conventions.

It is worth noting that a vehicle $p_i$ following the protocol will only grant another vehicle $p_j$ if $p_i$ predicts that $p_j$ will not interfere with $p_i$'s intended manoeuvre. This is an egoistic approach, since $p_i$ is not prepared to slow down for a lower priority vehicle to, for example, increase the overall throughput. Further more, $p_j$ will only ask $p_i$ if $p_i$ is in $p_j$'s membership -- $p_i$ has higher priority than $p_j$, where lower priority reflects which vehicle that according to the traffic rules should give-way. The combination of the egoistic approach to answering a request and the composition of the memberships therefore implies that the protocol in a failure free traffic situation should preserve priorities set by the traffic rules. This is also shown in the simulations, Figure \ref{sec:Results_TLPV}, where the priority vehicle is shown to slow down for the lower priority vehicle only if communication is lost during the request rounds.









\Section{Suggested extensions}

The Manoeuvre Negotiation Protocol was extended during this project and more extensions are possible. Some of the ones we have identified includes
\begin{itemize}
    \item \textbf{Allowing multiple granted vehicles}. The presented version of the protocol only allows a vehicle to grant one other vehicle at the time. A possible way of extending the protocol to allow multiple held grants without reducing safety is to divide the intersection into multiple critical sections instead of just one by defining conflicting manoeuvres, see Figure \ref{fig:conflicting_dir}. An agent receiving a request can thereby base its decision to grant or deny on if the intended manoeuvres overlap in time \emph{and} space. For example, two vehicles coming from opposite directions that both intend on doing a left turn through an intersection will have conflicting manoeuvres, while if they both intend on doing a right turn then their manoeuvres are not conflicting.
    
    \item \textbf{Extending to more traffic situations}. So far we have only tested the protocol on intersections and it would be interesting to see how well it performs in other traffic situations. It is possible to extend the priority matrix (Section \ref{sec:prioMatrix}) to be valid for roundabouts, merging lanes and possibly general roads to indicate how the traffic priorities are set up and how the memberships for the manoeuvre negotiation should be composed. An extension would also have to be made to the speed models to capture normal driving patterns in these situations.
    
    \item \textbf{Periodically changing the default priorities for fairness and deadlock avoidance}. The default priorities are static in our simulations to reflect the rules in an unsignalled intersection. An extention to signalled intersections could potentially be simulated by letting the default priorities change, indicating a change of priority road. This change would only have to be made by the membership service and would be made present to the vehicles by a changed membership. Changing the default priorities is also a way of increasing fairness and avoiding deadlocks for lower priority vehicles when there is a continuous flow of vehicles on the priority road.
\end{itemize}

\begin{figure}
    \centering
    \includegraphics[width=0.4\textwidth]{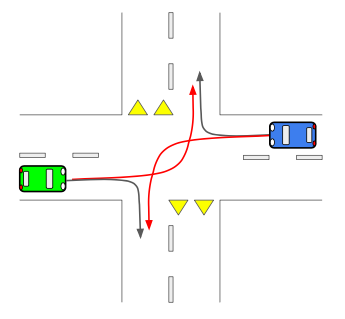}
    \caption{An example of conflicting manoeuvres in red -- both vehicles turning left -- and non-conflicting trajectories in black -- both vehicles turning right.}
    \label{fig:conflicting_dir}
\end{figure}

\Section{Suggested further testing}

\begin{itemize}
    \item \textbf{Varied turns and intersections.} There are more turn combinations than going straight and turning left that may cause collisions in an intersection. It would therefore be beneficial to run more tests with different combinations of manoeuvre intentions. There are also many more intersection types than just 4-way give-way intersections which also could be constructed in our simulation environment and tested.
    
    \item \textbf{Recorded real driving data.} Using real driving data instead of relying on speed models and fixed paths could give a better indication on our system's performance in real traffic scenarios. Recorded data from real vehicles in a 4-way, give-way intersection would then be needed to conduct the tests. The PID-controller on every vehicle in the simulation could be set to follow the recorded path and the recorded speed could function as a speed model.
    
    \item \textbf{Network simulator.} A network simulator could be used to simulate packet delays and message loss in a more realistic way. Packets sent in our setup in ROS are delivered instantly without any delay which would never happen in a real traffic scenario.
    
    \item \textbf{Multiple vehicles.} It could also be worth testing the system's performance when more then 2 vehicles are approaching the intersection, especially if a multi-grant version of the protocol is developed. Adding more vehicles to the simulation would make the intersection more busy and it would be interesting to evaluate the safety and throughput in this scenario as well.

\end{itemize}

\Section{Deadlock avoidance}


Deadlocks are possible in this system since the four conditions for causing deadlocks: mutual exclusion, resource holding, no preemption, and circular wait~\cite{Deadlocks} could be present in the system if the algorithm is not designed properly. Mutual exclusion is used for the whole (or in possible extensions: parts of) the intersection. Resource holding may occur since vehicles requesting for a permission may be required to get a grant from more than one vehicle and vehicles are only allowed to grant one vehicle at the time. No preemption is present if explicit release (section \ref{sec:explicit}) is used since the grant is held until a release message is received or the granted vehicle has left the intersection. Circular wait is possible if more than one vehicle are sending their requests concurrently since vehicles are only allowed to grant one vehicle at any time and grants from all expected respondents are required to get permission to perform the manoeuvre. 
As an example of a deadlock situation: Lets call two agents that concurrently send requests $p_i and p_j$. If $p_i$'s and $p_j$'s sets of requires responses have at least two agents in common, say $p_a$ and $p_b$ then $p_i$ may get a grant from $p_a$ while $p_j$ gets a grant from $p_b$ due to message delays. Agent $p_a$ then have to deny $p_j$ since only one grant can be given at any time, and the same applies to $p_b$'s response to $p_i$'s request. If $p_i$'s request is newer then $p_i$ will release its request and grant $p_j$ instead. But $p_i$ will also release its request since it will receive a deny from $p_b$ and the whole process starts over again.

Mutual exclusion for the intersection is vital for avoiding collisions so efforts should not be put there to remove that deadlock conditions. Getting a grant from every vehicle expected to answer is also a key feature to ensure safety, but it could potentially be possible to grant multiple vehicles and still avoid collisions. However, without multiple grants, what could be avoided is circular wait. If the vehicle with the newest request can be made to back off for a time long enough ($T_A + 2T_D$) for the vehicle with the newest request to get a grant then circular wait will not be possible anymore.

Intersections also have pre-defined priorities which can have an impact on formations of deadlocks. Using time stamps to set priorities of requests enables the avoidance of these deadlocks if the function $\noPriorityViolationDyn$ takes the pre-defined rules into account. A vehicle with a non-empty membership and that has not been fully granted should have the intention to stop and is not allowed to enter the intersection. This vehicle should therefore release its current request if it receives another request with an earlier time stamp.

\chapter{Conclusion}
\label{sec:conclusion}


In this thesis, a Manoeuvre Negotiation protocol intended for increasing the safety for autonomous vehicles have been presented.  A correctness proof was supplied where the requirements set on the protocol was shown to be fulfilled. The requirements stated limitation on the completeness of memberships, when manoeuvre negotiating requests can be initiated, how agents should answer to requests, when a grant should be released, and when an agent is allowed to enter the intersection.

Our simulations further show that the setups with the Manoeuvre Negotiation protocol (MN) used with or without the Risk Estimator (RE) are able to avoid collisions in more cases than when using merely RE. A higher priority vehicle equipped with a RE was shown to be able to detect all risky situation even in the presence additional noise, and communication loss. However, the detection of the risky situation was in many cases not made early enough for the vehicle to be able to break to avoid a collision.
The throughput for the higher priority vehicle when MN was activated was further shown to just marginally deviate from the ideal case which also implies that the protocol preserves the priorities set by the traffic rules. Our combined system thereby states an example of how vehicle-to-vehicle communication can be used alongside existing safety systems to increase safety for autonomous vehicles.

Finally, possibilities in ways to extend to the protocol and  additional testing methods were provided. Possible extensions included allowing vehicles to concurrently grant more than one vehicle at the time and broadening the scope in which the protocol can be used. Further testing is recommended to focus more on varying the driving patterns of the vehicle or even use reordered data from real vehicles to run the vehicles in the simulations.


\appendix

\chapter{An alternative implementation of \texorpdfstring{$\noPriorityViolationStat$}{Lg}}
\label{sec:nPV_Static_alt}

The purpose of $\noPriorityViolationStat$ is to determine if a priority violation can possibly occur between two vehicles when priorities are static. To draw the right conclusion, the algorithm has to consider both normal and extreme driving behaviour, which can be done with the help of occupancy prediction~\cite{Occupancy}. The main idea of occupancy prediction is to, given a start state of an agent, provide an upper limit to the agent's position some time later using constraints on vehicle motion.


Pseudocode for $\noPriorityViolationStat$ is given in \ref{alg:nPVStatic}. The procedure takes four input parameters including the agent states of two vehicles, $p_i$ and $p_j$, the time from which calculations should begin, $t$, and a time step $\delta_t$. The states $p_i$ and $p_j$ are used to compute the initial occupancy area $\mathcal{O}_i^0$ and $\mathcal{O}_j^0$ at the time $t$.

An occupancy area $\mathcal{O}_i^k$ covers every possible position the agent represented by $p_i$ can have in the time interval for which it was computed. Both physical and legal constraints, such as limited engine power and speed limits~\cite{Occupancy}, can be applied to predict $\mathcal{O}_i^k$ from $\mathcal{O}_i^{k-1}$. The pseudocode uses the interface $getOccupancyArea$ to calculate these areas.

The implementation approach for $\noPriorityViolationStat$ considers the intersection of the occupancy areas $\mathcal{O}_i^k \cap \mathcal{O}_j^k$ for consecutive time intervals $\Delta t_k = [t + k \delta_t, t + (k+1) \delta_t]$, where $k = 0,1,2 ...$. Traffic priorities are preserved by letting an agent $p_k$ grant a request from another agent $p_j$ if, and only if, $\mathcal{O}_i^k \cap \mathcal{O}_j^k = \emptyset \forall k \in [0,M] : ((hasLeftIntersection(\mathcal{O}_i^M) \lor hasLeftIntersection(\mathcal{O}_j^M)) \land (\neg hasLeftIntersection(\mathcal{O}_i^{M-1}) \land \neg hasLeftIntersection(\mathcal{O}_j^{M-1}))$, that is, $\Delta t_M$ is the first time interval after which $p_i$ or $p_j$ is guaranteed to have left the intersection. 


Further constraints in the predicted vehicle motions can be applied if multiple vehicles are present in a traffic scene~\cite{IntentionAwareOccupancy}. Let $\{p_i\}_{i=0^N}$ be $N$ vehicles in a queue on a one-lane road where $p_1$ is the leading vehicle. The occupation areas $A_j(p_i) \forall i \in [2,M]$ will then be bounded in the forward direction by $A_j(p_1)$ at any time instant j as long as the road does not expand to a multi-lane road. 

The presented approach should be able to  predict all possible priority violations caused by \emph{predictable} vehicle motion since the approach relies on the physical constraints of the vehicle. We define predictable vehicle motion as all kinds of driving in the absence of rare natural events such as appearing sink holes, rockfall, sudden flooding etc, which can cause the vehicle to move in unexpected ways.

This version of $\noPriorityViolation$ was selected to be used by the membership service to ensure that a membership for an agent $p_i$ contains all other vehicles that could potentially cause a priority violation with $p_i$.
However, this approach is not efficient regarding throughput in normal traffic situations where most vehicles follow normal trajectory and speed patterns. We therefore decided to employ another version of $\noPriorityViolation$ based on behaviour models and intersection occupancy intervals for making the grant or deny decision. This version is presented in Section \ref{sec:nPV_Dynamic}.

\begin{algorithm}[ht!]
	\begin{\VCalgSize}
		\textbf{Structures:}\\
		$AgentState = (ta,x,v,a)$\tcp*[r]{timestamp, position, velocity, acceleration}
		
		\ \\
		\textbf{Interfaces:}\\
	    $getOccupancyArea(t_s, t_f, p)$: \;
	    $hasLeftIntersection(\mathcal{O}(p))$: returns $\true$ if all predicted states of p, $\mathcal{O}(p)$, are outside and moving away from the intersection\; 

        \ \\
		\textbf{procedure} $\noPriorityViolationStat(p_i, p_j, t, \delta_t)$ \Begin{
		    \textbf{let} $t_c = t$\;
		    \textbf{let} $\mathcal{O}_{i} = getOccupancyArea(p_i.ts, t, p_i)$\;
		    \textbf{let} $\mathcal{O}_{j} = getOccupancyArea(p_j.ts, t, p_j)$\;
		    \textbf{while} $\neg(hasLeftIntersection(\mathcal{O}_{i}) \lor hasLeftIntersection(\mathcal{O}_j))$ \textbf{do} \\ 
		    \hspace{1em} \textbf{let} $\mathcal{O}_{i} = getOccupancyArea(t_c,t_c+\delta_t,\mathcal{O}_{i})$\;
		    \hspace{1em} \textbf{let} $\mathcal{O}_{j} = getOccupancyArea(t_c,t_c+\delta_t, \mathcal{O}_{j})$\;
		    \hspace{1em} \lIf{$\mathcal{O}_{i} \cap \mathcal{O}_{j} \neq \emptyset$}{\textbf{return} $\false$}
		    \hspace{1em} \textbf{let} $t_c = t_c+\delta_t$\;
			{\textbf{return} $\true$}\;
		}
	\end{\VCalgSize}
	\caption{$\noPriorityViolationStat$ relying on static priorities.}
	\label{alg:nPVStatic}
\end{algorithm}

\end{document}